\documentclass[pdflatex,sn-mathphys-num]{sn-jnl}

\usepackage{graphicx}%
\usepackage{multirow}%
\usepackage{amsmath,amssymb,amsfonts}%
\usepackage{amsthm}%
\usepackage{mathabx}
\usepackage{mathrsfs, mathtools}%
\usepackage[title]{appendix}%
\usepackage{caption, subcaption}
\usepackage{enumerate}
\usepackage{xcolor, xspace}%
\usepackage{csquotes}%
\usepackage{cleveref, hyperref}

\newcounter{counter} \numberwithin{counter}{section}
\theoremstyle{thmstyleone}%
\newtheorem{theorem}[counter]{Theorem}
\newtheorem{proposition}[counter]{Proposition}
\newtheorem{lemma}[counter]{Lemma}
\newtheorem{corollary}[counter]{Corollary}
\theoremstyle{definition}

\theoremstyle{thmstyletwo}%

\theoremstyle{thmstylethree}%
\newtheorem{definition}[counter]{Definition}
\newtheorem{assumption}[counter]{Assumption}

\DeclarePairedDelimiter\parens{\lparen}{\rparen}

\DeclarePairedDelimiter\bracks{\lbrack}{\rbrack}
\DeclarePairedDelimiter\angles{\langle}{\rangle}

\newcommand{\E}{\mathbb{E}}
\def\P#1{\mathrm{\mathbb{P}}\left[#1\right]}
\newcommand{\R}{\mathbb{R}}

\newcommand{\calS}{\mathcal{S}}
\newcommand{\Pois}{\operatorname{Pois}}
\newcommand{\PPrio}{\operatorname{P-Prio}}
\newcommand{\wideW}{\widehat{W}}
\newcommand{\PoisCost}{\Phi}
\newcommand{\costbar}{\widebar{Cost}}

\newcommand{\obj}{\widehat{Cost}}
\newcommand{\priotail}[1]{\P{T_1^{\PPrio(2;1)}>#1}}
\newcommand{\lookaheadtail}[1]{\P{T_1^{\alpha}>#1}}

\newcommand{\BP}{\operatorname{BP}}
\newcommand{\Exp}{\operatorname{Exp}}
\newcommand{\Overtake}{\operatorname{Overtake}}
\newcommand{\LookAhead}{LookAhead\xspace}

\newcommand{\eff}{\operatorname{eff}}

\newcommand{\edit}[1]{\color{black}{#1}\color{black}}
\begin{document}

\title[Article Title]{\LookAhead: The Optimal Non-decreasing Index Policy for a Time-Varying Holding Cost problem}


\author*[1]{\fnm{Keerthana} \sur{Gurushankar}}\email{kgurusha@andrew.cmu.edu}
\equalcont{These authors contributed equally to this work.}

\author[1]{\fnm{Zhouzi} \sur{Li}}\email{zhouzil@andrew.cmu.edu}
\equalcont{These authors contributed equally to this work.}

\author[1]{\fnm{Mor} \sur{Harchol-Balter}}\email{harchol@cs.cmu.edu}

\author[1]{\fnm{Alan}\sur{Scheller-Wolf}\email{awolf@andrew.cmu.edu}}

\affil[1]{\orgname{Carnegie Mellon University}, \orgaddress{\street{5000 Forbes Avenue}, \city{Pittsburgh}, \postcode{15213}, \state{PA}, \country{USA}}}


\abstract{
In practice, the cost of delaying a job can grow as the job waits. Such behavior is modeled by the Time-Varying Holding Cost (TVHC) problem, where each job's instantaneous holding cost increases with its current age (a job's age is the time since it arrived). The goal of the TVHC problem is to find a scheduling policy that minimizes the time-average total holding cost across all jobs.

However, no optimality results are known for the TVHC problem outside of the asymptotic regime. In this paper, 
we study a simple yet still challenging special case: A two-class M/M/1 queue in which class 1 jobs incur a non-decreasing, time-varying holding cost and class 2 jobs incur a constant holding cost.

Our main contribution is deriving the first optimal (non-decreasing) index policy for this special case of the TVHC problem. Our optimal policy, called \LookAhead, stems from the following idea: Rather than considering each job's {\em current} holding cost when making scheduling decisions, we should look at their cost some $X$ time into the future, where this $X$ is intuitively called the ``lookahead amount."  This paper derives that optimal lookahead amount.
}





\maketitle

\section{Introduction}\label{sec1}

Every day, systems around us must decide which task to do next. A clinic decides which patient to see, a data center chooses which request to serve, a factory picks which order to process. In many of these settings, the cost of delaying a job increases as it waits. In a hospital emergency department, triage protocols must account for the fact that a patient with a moderate condition may become critical if not treated in time. In a cloud platform delivering live video or interactive applications, delays past a few hundred milliseconds can cause perceptible lag. In warehouses, orders with perishable goods may incur increasing holding costs due to spoilage risk or expiry windows. 
In all these systems, the scheduler must balance limited service capacity against growing costs of delay, and must often trade off between completing a short but low-urgency job and an expensive job that may take longer. 

All these problems can be modeled as a stream of jobs where each job has some instantaneous holding cost (cost per second that the job is not complete). A job's holding cost may be constant or vary over time. At every moment, the service provider has to pay a total holding cost across all jobs in the system. The goal of the service provider is to minimize the time average holding cost across all jobs. 

More formally, consider the problem of scheduling jobs in a single-server queue to minimize the time-average holding cost. This is a classical objective in queueing theory and operations research. Optimal scheduling is well understood when jobs incur constant instantaneous holding cost while in system. For example, when job sizes are exponential, the well known $c\mu$ rule 
is optimal~\cite{cox2020queues, Buyukkoc_Varaiya_Walrand_1985}. For many extensions such as general job size distributions,  optimal scheduling is given by the Gittins index policy~\cite{gittins1983dynamic,scully2021gittins}. 

However, when the instantaneous holding cost varies over time, no optimal policy is known. 
To capture cases where holding costs increase as a job is waiting, we define a job's \emph{age} as the time that  the job has spent in the system, and we let the instantaneous holding cost of a job be a function of its age. 
This {\em Time-Varying Holding Cost (TVHC)}  regime was first studied in the seminal paper by Van Mieghem~\cite{van1995dynamic}.  That paper introduces the generalized c$\mu$ rule to minimize time-average holding cost.  However, the rule turns out to only be optimal in a special asymptotic regime, and so far no optimality results are known outside of asymptotic regimes (see Section~\ref{sec:related-works}). 

Note that our definition of {\em age}, based on time in system, is in contrast to \emph{attained service} which is sometimes called age in other literature~\cite{scully2018soap,nuyens2008foreground}. 
The key difficulty with the fact that holding costs increase with age (not just attained service) is that the problem intrinsically becomes a \emph{restless multi-armed bandit} (R-MAB) problem. These problems are notoriously difficult; it is well known that the optimal policy for such problems is usually intractable~\cite{guha2010approximation}. 


\begin{figure}
    \centering
    \includegraphics[width=0.5\linewidth]{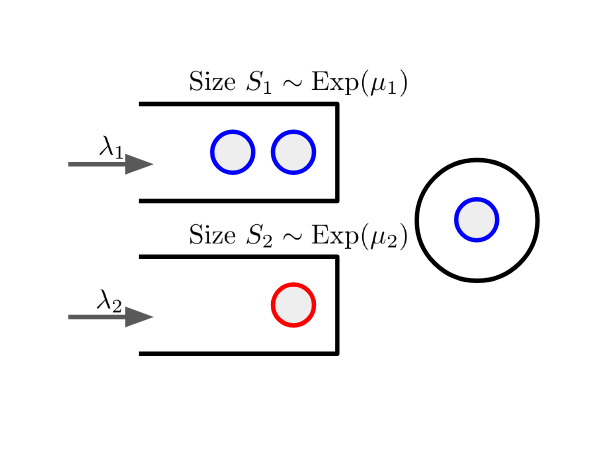}
    \caption{A 2-class M/M/1 Queue with Age-based Holding Costs.}
    \label{fig:2class_MM1}
\end{figure}

\subsection{Our problem}
In this paper, we study a simple setting of the TVHC problem (see \Cref{fig:2class_MM1,fig:holding}) which still captures the key difficulty of restlessness: A 2-class M/M/1 queue in which one class of jobs has a non-decreasing holding cost function that depends on the job's age (time in system), while the other class has a constant instantaneous holding cost. Jobs have exponential service times with class-dependent rates. 
Even this simple sounding problem is entirely open. Our model captures many settings in which one class has a dominant delay sensitivity. For instance, user facing jobs in computer systems serving mixed workloads alongside background processes may have important and complex latency requirements. Likewise in business operations, corrective maintenance tasks may be much more urgent than preventative tasks. Our problem models many such settings using a general non-decreasing holding cost function for one class, and a constant holding cost for the other.

\begin{figure}[h]
    \centering  
    \begin{subfigure}{0.48\textwidth}  
        \centering
        \includegraphics[width=0.9\textwidth]{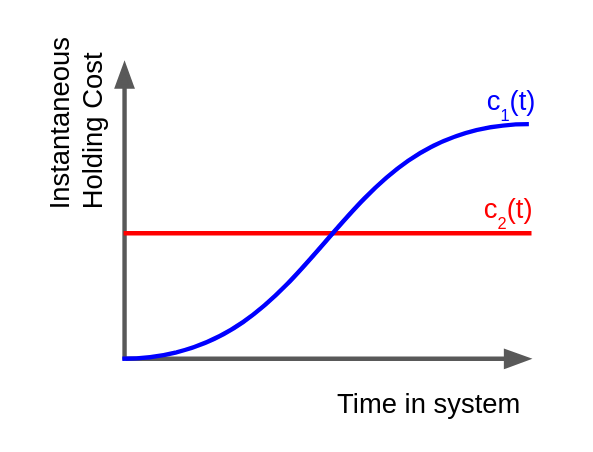}
        \caption{Instantaneous holding cost}
        \label{fig:ins holding}
    \end{subfigure}
    \hfill 
    \begin{subfigure}{0.48\textwidth} 
        \centering
        \includegraphics[width=0.9\textwidth]{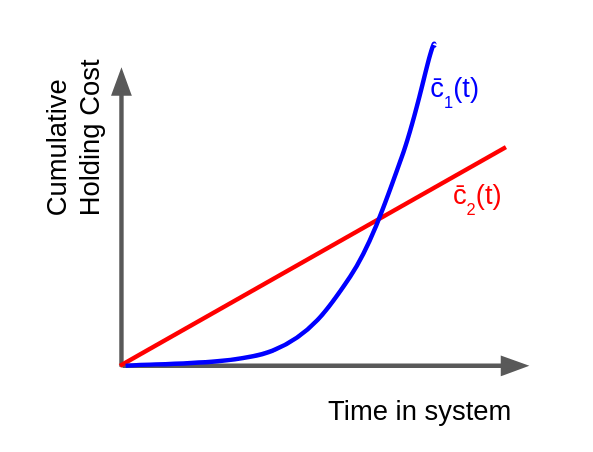}
        \caption{Cumulative holding cost}
        \label{fig:cum holding}
    \end{subfigure}
    \caption{(Instantaneous) Holding cost and Cumulative Holding Cost.}
    \label{fig:holding}
\end{figure}

In general, the optimal policy for our problem may be arbitrarily complex. It could be viewed as the solution to a Markov Decision Process (MDP) that needs to track every combination of jobs with every age. One way of simplifying the problem is to limit the search space to \emph{index policies}. In an index policy, every job is assigned a scalar-valued {\em index} as a function of its state (class and age)\footnote{If job sizes are not exponentially distributed, then the index can also be a function of the job's attained service, in addition to its age.}, and the scheduler always serves the job with highest index. The vast majority of scheduling policies studied in the queueing literature are index policies (see e.g.~\cite{nuyens2008preventing,scully2018soap,harchol2013performance,scully2021gittins,rai2004performance, Sigmetrics05b,aalto2024whittle,anand2018whittle}).
Index policies are also known to be optimal in many settings, such as the Gittins index for rested bandits.  However, in our setting with age-dependent holding costs, even the optimal policy among the class of index policies is unknown. 

We observe that when a class's holding cost increases with age  and its job sizes are exponential with unknown remaining service time, the optimal policy must serve jobs within that class in \emph{first-come-first-serve} (FCFS) order~\cite[Lemma 3.1]{li2025improving}. This motivates us to restrict our search to the class of index policies which enforce FCFS within each class: Equivalently, each class has a non-decreasing index function and ties between jobs within a class are broken using FCFS. We refer to such index policies as the set of {\em non-decreasing index policies} in our paper, and our goal is to find the optimal policy within that set. 



\subsection{Towards solving our problem}

We now turn to the question of finding the optimal non-decreasing index policy. Intuitively, we want something like the {\em $c \mu$ rule}, where the index of a job with holding cost $c_i$ and completion rate $\mu_i$ is $c_i \cdot \mu_i$.  While the $c \mu$ rule is optimal for the case of constant holding costs, it is not defined for our setting. When holding costs are not constant, the {\em generalized $c\mu$ rule \cite{van1995dynamic}} assigns index $c_i(t) \cdot \mu_i$ to a job of class $i$ and age $t$. Unfortunately, the generalized $c \mu$ rule is only asymptotically optimal for the TVHC problem.  

It is easy to see what goes wrong in the generalized $c \mu$ rule: If we look at Figure~\ref{fig:holding} when $\mu_1=\mu_2$, a job of class 1 only gets priority over class 2 at the {\em instant} that the blue and red curves cross.  
Would it not make more sense to prioritize class 1 jobs a little earlier, so that we can get those jobs done {\em before} their holding cost gets really high?
That is exactly the intuition behind our proposed index policy, which we call {\em \LookAhead}.  Under \LookAhead, the index of class 1 jobs depends not on their holding cost at time $t$, but rather their holding cost in the future, specifically $X$ time into the future.  The ``lookahead amount," $X$, is defined in Theorem~\ref{thm:Lookahead}, below, and intuition for the amount is given later in the paper.

\color{black}

\begin{theorem}[\LookAhead]\label{thm:lookahead}
    Suppose we have a 2-class M/M/1 queue, where class $i$ jobs arrive at rate $\lambda_i$ with service rate $\mu_i$ for $i\in\{1,2\}$. Class $1$ jobs incur instantaneous holding cost at rate $c_1(t)$ when they have age $t$, and class $2$ jobs incur holding cost at constant rate $c_2$ while in system. 
    The optimal policy among the class of non-decreasing index policies is given by the following index functions:
    \[
        V_1(t) = \mu_1\E[c_1(t+X)],  \qquad \text{where }X\sim\Exp(\mu_1-\lambda_1), \qquad \text{ and }V_2(t) = \mu_2c_2.
    \]
    \label{thm:Lookahead}
\end{theorem}

Ultimately, our paper makes the following novel contributions:
\begin{enumerate}[(i)]
    \item We derive the first optimal scheduling result for a TVHC problem.
    \item We derive the first analysis of response time tail for an age-based scheduling policy (see \Cref{lemma: tail prob}).
    \item Our index policy coincides with the Whittle index policy for the same R-MAB as we define~\cite{li2025improving}. The Whittle index policy is typically a heuristic with only asymptotic optimality guarantees~\cite{weber1990index,verloop2016asymptotically,gast2023exponential}. Thus we show a surprising case of optimality for a Whittle index policy. 
\end{enumerate}

\edit{Thus we establish index-optimal scheduling for a specific two-class TVHC setting. Nonetheless, the amortized $c\mu$ interpretation presented in \Cref{sec:amortized} as well as our recent work on the Whittle index for TVHC~\cite{li2025improving} suggest principled heuristic policies for broader TVHC problems. In non-preemptive regimes, the same index function can be applied at departure instants. An important question which is outside the scope of this paper is to characterize the gap between index-optimal scheduling and truly optimal scheduling. Because the globally optimal policy is defined over a high-dimensional continuous-time state space, even empirical evaluation of this gap poses significant challenges.}

\section{Related Work}\label{sec:related-works}

In this section, we review related work on scheduling that informs our study of time-varying holding costs (TVHC). We group prior work into four areas:

First, we review classic scheduling policies that prioritize jobs based on service received so far (\Cref{sec:servicebased}). These have been well-studied but cannot analyze age-based scheduling. Next, we discuss age-based scheduling policies (\Cref{sec:agebased}). These are closer to our setting but only consider a few simple policies, and focus on response time analysis rather than optimal scheduling. 
We then describe the generalized $c\mu$ rule (\Cref{sec:cmu}), the only prior work offering any optimality result for TVHC; though only in the diffusion limit. Finally, we examine Whittle index-based heuristics (\Cref{sec:whittle}) for TVHC. 


\subsection{Prioritizing jobs based on attained service}\label{sec:servicebased}
The past half century of queueing literature has devoted enormous energy to studying the response time of various scheduling policies in the M/G/1 queue, where the scheduling policy makes decisions based on the {\em service that a job has received so far} 
\cite[Chapters 28-33]{harchol2013performance}.
Examples include \emph{Shortest Remaining Processing Time (SRPT)} and \emph{Least Attained Service (LAS)} \cite{rai2004performance}. Recently (2018), Scully \textit{et al.}~\cite{scully2018soap} proposed a unified analysis framework (SOAP) for a very broad class that includes almost all known  scheduling policies, in which a job’s priority depends on its own class, size and the service that it has received so far.

\subsection{Prioritizing jobs based on age}\label{sec:agebased}

Much less research has gone into the analysis of scheduling policies where a job's priority is based on the time it has spent in the system.  
The notable exception is \emph{accumulating priority scheduling}~\cite{stanford2014waiting,fajardo2017waiting,mojalal2020lower}.  Under this model, a job's priority grows linearly with the time it spends in the system (starting at priority 0).  
Different classes may accumulate priority at different rates, allowing an older slow-accumulating job to eventually overtake a newer fast-accumulating one.

However, jobs in accumulating priority models are not associated with a holding cost function, and the goal is not to minimize holding cost, but rather to analyze response time. 
Additionally, accumulating priority policies only allow analysis in settings where two jobs that arrive at the same time will never flip relative priority~\cite{cildoz2019accumulating,mojalal2020lower}. 
That said, we do borrow some of their Poisson process transformation lemmas, in our own technical arguments (such as adapting their \cite[Lemma 4.2]{stanford2014waiting} for our \Cref{lemma:Q0}). 


\subsection{Asymptotically-Optimal Scheduling for the TVHC problem:  Generalized $c \mu$}
\label{sec:cmu}


The TVHC problem we consider was introduced in van Mieghem's seminal paper~\cite{van1995dynamic}. This paper provided a scheduling policy known as the \emph{generalized $c\mu$ rule} for minimizing age-based convex holding costs. 
The generalized $c\mu$ rule indexes jobs by the product of their instantaneous holding cost and instantaneous service rate. This policy was shown to be asymptotically optimal for minimizing time-average total holding cost in the diffusion limit (i.e. when arrival and service rates scale as $n\lambda$ and $n\mu$ as $n\to \infty$) under heavy traffic for multi-class M/G/1 queues with convex holding costs as a function of delay~\cite[Proposition 8]{van1995dynamic}.  

The generalized $c\mu$ rule is, to our knowledge, the only work to provide \emph{any} optimality guarantees for TVHC, and only in the asymptotic scaling regimes. In contrast, our work develops an index policy that is provably optimal in finite time-scales.

Interestingly, the generalized $c\mu$ rule corresponds to our policy in the special case where the optimal lookahead parameter $X\to 0$. This happens precisely in the diffusion limit, thus our result is consistent with van Mieghem's Proposition 8, and hence is also asymptotically optimal.
But under finite time-scales, we demonstrate (see Section~\ref{sec:simulations}) that our \LookAhead policy can significantly outperform the generalized $c\mu$ rule.

\subsection{Heuristic Scheduling for TVHC:  Whittle}\label{sec:whittle}





Many scheduling problems can be viewed as instances of \emph{multi-armed bandit} (MAB) problems. In a standard Markov MAB model, each arm corresponds to a Markov process with internal states that evolve according to a transition model. The agent selects one arm to activate at each time step, accruing a cost determined by the states of all arms, and seeks to minimize the long-run average cost. In scheduling, each job or job class corresponds to an arm, and the decision is which job to serve at a given time.

The MAB framework has led to powerful results in optimal scheduling, such as the \emph{Gittins index policy}, which is optimal for preemptive scheduling in an M/G/1 queue with unknown job size distributions~\cite{scully2021gittins}. These results fall under the category of \emph{rested bandits}, where the state of an arm (i.e., a job) evolves only when the arm is pulled—equivalent to a job being served.

However, when a job's state (e.g. age) 
evolves even while it is not being served, the problem becomes a \emph{restless bandit}. This setting is significantly harder: It is known to be PSPACE-hard in general, and optimal policies are difficult to characterize~\cite{guha2010approximation}. 
Several papers have proposed Whittle index-based policies for various scheduling problems~\cite{anand2018whittle,ansell2003whittle,li2025improving,aalto2024whittle}. However, works in this literature only provide heuristics with only asymptotic guarantees. 


In particular, Li et al.~\cite{li2025improving} provide a Whittle-index policy for TVHC. This work develops a heuristic index policy based on job age and demonstrates strong empirical performance. In fact, the policy that we analyze in this paper coincides with the one proposed in \cite{li2025improving}. However, \cite{li2025improving} does not establish optimality. In contrast, we provide a rigorous proof that this policy is the optimal index policy for one type of two-class TVHC problem.

We also borrow several key structural lemmas from \cite{li2025improving}, particularly in our analysis of the cost dynamics and structural monotonicity properties of the value function, which are used in our main proof arguments. To the best of our knowledge, ours is the first work to prove the optimality of an exact index policy for a dynamic, delay-based holding cost problem.

For a broader overview of bandit models and index policies in scheduling, we refer readers to \cite{nino2023markovian,li2025improving}.

\section{A Shorter Proof Using an Additional Assumption}\label{sec:shorter-proof}
In this section, we present a shorter proof of our result using one additional assumption. Intuitively, since class 2 jobs have constant holding cost, it is reasonable to assume that the index of class 2 jobs should be constant. This is the crux of Assumption~\ref{assumption: class2 constant}   (explained more rigorously below). Based on this assumption, we will now derive the optimal non-decreasing index policy. A longer and more involved proof without this assumption is provided in Section~\ref{sec:longer-proof}.

We structure this section as follows: In \Cref{sec:assumption}, we define Assumption~\ref{assumption: class2 constant}.   
In \Cref{sec:lookahead}, 
we discuss the structure imposed on the scheduling problem when class 1 has an arbitrary non-decreasing index function and class 2 has a constant index. We call the class of policies defined by this structure ``$\Overtake$.'' In \Cref{sec:main technique}, we analyze the response time tail for any $\Overtake$ policy. In \Cref{sec: optimization problem}, we use the response time tail analysis to analytically solve the optimization problem of minimizing time average holding cost among $\Overtake$ policies.

\subsection{Discussion of the Additional Assumption}
\label{sec:assumption}

Recall that the goal of this paper is to find the optimal non-decreasing index policy. This section makes an assumption that the index for class 2 jobs should not only be non-decreasing, but furthermore be a constant. 


\begin{assumption}
\label{assumption: class2 constant}
There exists an optimal non-decreasing index policy where the index for any class 2 job is a constant regardless of the job's age.
\end{assumption}

Although this assumption is well-motivated, it is not obvious. For example, one may argue that the age of a class 2 job reveals some information about the number of class 2 jobs in the system, which may affect the decision.  The longer proof in Section~\ref{sec:longer-proof} addresses these concerns and justifies our result.

Using Assumption~\ref{assumption: class2 constant} allows us to restrict our optimization problem to the class of  $\Overtake$ policies, which provides for a simpler proof.

\subsection{Policy class: $\Overtake$}\label{sec:lookahead}
For any non-decreasing index policy where the class 2 index is a constant, let $V_1, V_2$ be the index functions of class 1 and class 2 jobs respectively.
Define $\alpha$ to be the youngest age where $V_1(\alpha)=v=V_2(\cdot)$, where $v$ represents the constant index of any class 2 job. 
Then
the index policy is equivalent to a policy which has three levels of priority: ($Q_0$) class 1 jobs with age  $\geq \alpha$; ($Q_2$) all class 2 jobs; and ($Q_1$) class 1 jobs with age  $< \alpha$. Within each level of priority, jobs are served in FCFS order.
 We call this policy {\em $\Overtake(\alpha)$}. Pictorially, it can be illustrated in Figure~\ref{fig:Lookahead}, where priority is enforced as: $Q_0>Q_2>Q_1$.    

\begin{figure}[h]
    \centering 
    \scalebox{1}{ 
    \includegraphics[width=0.40\textwidth]{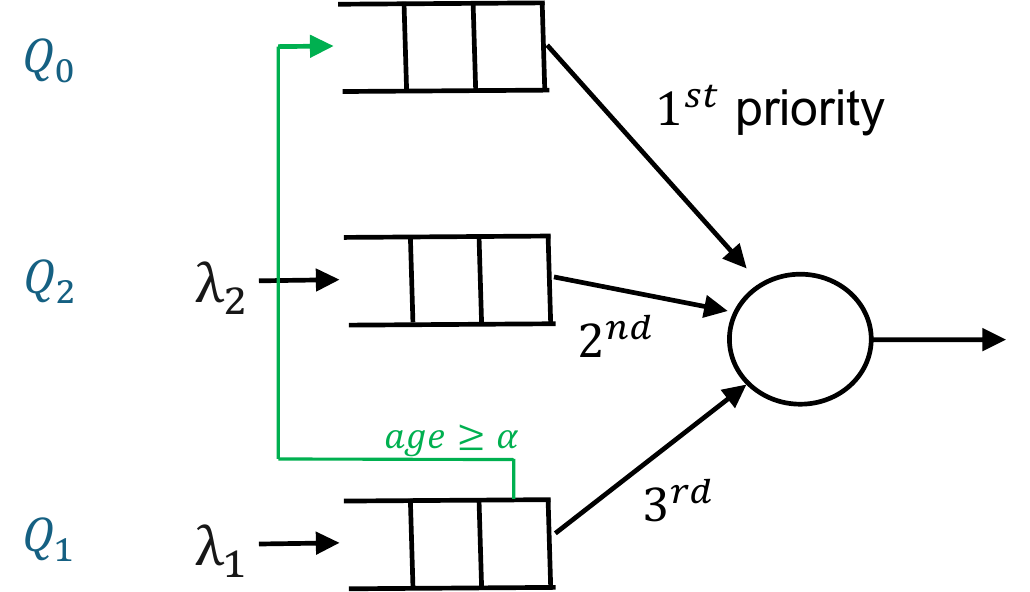}}
    \caption{$\Overtake(\alpha)$ Policy. }
    \label{fig:Lookahead}
\end{figure}

In this way, any non-decreasing index policy with constant class 2 index is equivalent to some $\Overtake$ policy. Thus, by Assumption~\ref{assumption: class2 constant}, the optimal non-decreasing index policy is equivalent to some $\Overtake$ policy. Therefore, to find an optimal non-decreasing index policy, it suffices 
to optimize time-average total holding cost within the class of $\Overtake$ policies.  


\subsection{Characterizing Response Time Tail}\label{sec:main technique}
In this section, we derive the tail of the response time of class 1 jobs under $\Overtake(\alpha)$. The key insight is the following lemma: The busy period of $Q_0$ can be seen as that of an M/M/1 queue. Note that this does not imply that $Q_0$ is an M/M/1 queue, because the inter-arrival times of jobs into $Q_0$ during a $Q_0$-idle period are not exponentially distributed. 
\begin{lemma}
\label{lemma:Q0}
Let $T^{M/M/1}$ denote the response time of an M/M/1 queue with arrival rate $\lambda_1$ and job size distribution $\Exp(\mu_1)$ under an FCFS policy. Then jobs that enter $Q_0$ will stay in $Q_0$ for a time distributed as $T^{M/M/1}$.
\end{lemma}
\begin{proof}
    We track the dynamics of the oldest age in $Q_0$ and in a standard M/M/1 queue. Note that for any sample path, the oldest age trajectory uniquely determines the response time of each job (since each oldest age drop is equivalent to a job completion). Thus, it suffices to prove that the stochastic behaviors of the oldest age in $Q_0$ and the M/M/1 queue are the same during a $Q_0$ busy period.

    

    Define the oldest age in an M/M/1 queue to be $A$. If $A<0$, this means there is no job in the M/M/1 queue, and the next job arrives $(-A)$ time later. In this case, $A$ grows with rate 1. Otherwise if $A>0$,  during each $\delta$ time step, with probability $1- \mu_1 \delta+o(\delta)$, no job is finished and the oldest age grows with rate 1. With probability $\mu_1\delta + o(\delta),$ 
    the oldest job is finished.\footnote{For the remainder of this paper, we use $\delta$ to refer to a small time step and omit the $o(\delta)$ terms in $\delta$-time step discrete optimizations of our continuous time system.} \footnote{An alternative way to describe the dynamics of $A$ is $dA = dt - I \cdot dN_{\mu_1}(t),$
    where $ I\sim \Exp(\lambda_1)$, and $N_{\mu_1}  \text{ is a Poisson counting process with rate $\mu_1$.}$} In this case, since the inter-arrival time is $\Exp(\lambda_1)$, we know that the oldest age drops by a random variable $\Exp(\lambda_1)$. If the oldest age drops below 0, the busy period ends.
    
    Similarly, in $Q_0$, if there exists jobs in $Q_0$, define the time the oldest job in $Q_0$ has spent in $Q_0$ to be $A_0$. 
    By definition of the $\Overtake(\alpha)$ policy, $A_0$ is equal to the age of the oldest job in $Q_0$ minus $\alpha$. Thus, 
    a busy period in $Q_0$ starts with $A_0=0$, and then the oldest age $A_0$ follows the same dynamics as $A$ in an M/M/1 queue (when $A>0$), until $A_0$ drops below $0$, which means that the $Q_0$ busy period ends. 
    Therefore, $A_0$ has the same dynamics as $A$ when they are larger than 0, which means the job completion process in $Q_0$ is the same as that in an M/M/1 queue during any $Q_0$ busy period. Thus, the time a job spends in $Q_0$ given it has entered $Q_0$ is equal to the response time in an M/M/1 queue. 
\end{proof}

The following lemma characterizes the tail probability of a class 1 job's response time. Using Lemma~\ref{lemma:Q0}, we are able to express a class 1 job's response time in terms of the response time under a strict priority queueing system where class 2 has preemptive priority over class 1. 

\begin{lemma}[Tail Probability of Class 1 jobs]
\label{lemma: tail prob}
Let $T_1^\alpha$ denote the response time of class 1 jobs under the $\Overtake(\alpha)$ policy. Let $T_1^{\PPrio(2;1)}$ denote the response time of class 1 jobs in a preemptive priority system where class 2 jobs  have strict priority over class 1 jobs. Then
   the tail probability of $T_1^\alpha$ is as follows:
    \begin{equation}
    \lookaheadtail{t}= 
    \begin{cases} 
        \priotail{t} & \text{if } t \leq \alpha, \\
        \priotail{\alpha}\cdot e^{-(t-\alpha)(\mu_1-\lambda_1)} & \text{otherwise}. 
    \end{cases}
    \label{eq:lookaheadtail}
    \end{equation}
\end{lemma}

\begin{proof}
We start by considering the case where $t\leq \alpha$.  
We will prove that \[\P{T_1^{\alpha} \leq t} = \P{T_1^{\PPrio(2;1)}\leq t}.\]

Let  $t \leq \alpha$.  For any tagged class 1 job whose response time is $\leq t$, let $W^\alpha$ denote the work that runs ahead of the tagged class 1 job  in the $\Overtake(\alpha)$ system.  Because our tagged class 1 job does not enter $Q_0$, we can express $W^\alpha$ as a sum of two components: (i) all the class 1 jobs which arrive before our job, and (ii) all the class 2 jobs which arrive before our tagged class 1 job completes. Therefore $W^\alpha$ is the same as $W^{\PPrio(2;1)}$, the work that runs ahead of our tagged class 1 job in a $\PPrio(2;1)$ system. Therefore, the response time of our tagged class 1 job is the same under $\Overtake(\alpha)$ and under $\PPrio(2;1)$, assuming that these response times are $\leq \alpha$.   We have thus shown that 
\[\P{T_1^{\alpha} \leq t} \leq \P{T_1^{\PPrio(2;1)}\leq t}.\]

Moreover, the response time of a class 1 job is smaller in any $\Overtake$ policy than in $\PPrio(2;1)$. Thus  
\[\P{T_1^{\alpha} \leq t} \geq \P{T_1^{\PPrio(2;1)}\leq t}.\]
These two facts show that 
\[\lookaheadtail{t} = \priotail{t},\quad \text{for }t\leq \alpha.\]

For $t>\alpha$, we can condition on the class 1 job reaching age $\alpha$ (thus entering $Q_0$):
\begin{align*}
    \lookaheadtail{t}&= \P{T_1^{\alpha} > t \mid T_1^{\alpha}>\alpha}\cdot\lookaheadtail{\alpha}\\
    &= \P{T_0^{\alpha} > t-\alpha \mid T_1^\alpha>\alpha}\cdot\priotail{\alpha},
\end{align*}
where $T_0^\alpha$ is the time that class 1 job stays in $Q_0$.

By Lemma~\ref{lemma:Q0}, we know that $T_0^\alpha$ is the same as the response time in an M/M/1 queue. It is well-known that the tail probability of response time in this M/M/1 queue is 
\[\P{T^{M/M/1}>t} = e^{-(\mu_1 - \lambda_1)t}.\]
Thus we have that 
\begin{equation}
    \P{T_0^{\alpha} > t-\alpha \mid T_1^{\alpha}>\alpha} = Pr[T^{M/M/1}>t-\alpha]= e^{-(\mu_1-\lambda_1)(t-\alpha)}.
    \label{eq: T0}
\end{equation}

This allows us to express the tail probability when $t>\alpha$ as:
\begin{align*}
\lookaheadtail{t} &= \P{T_0^{\alpha} > t-\alpha \mid T_1^\alpha>\alpha}\cdot\priotail{\alpha}\\
&= \priotail{\alpha}\cdot e^{-(\mu_1-\lambda_1)(t-\alpha)}.    
\end{align*}

\end{proof}

\subsection{Optimization Problem}\label{sec: optimization problem}

In this section, we derive the optimal $\Overtake$ policy. The proof is organized in the following order: First, 
using the conservation law on mean response times~\cite{kleinrock1965conservation}, we transform our optimization problem into terms purely dependent on $T_1^{\alpha}$ in \Cref{lemma:obj}. We then use our tail analysis from \Cref{lemma: tail prob} to analytically solve the resulting optimization problem: Two lemmas are derived in \Cref{lemma:exp formula,lem:lookahead-ECT} and finally we prove \Cref{thm:opt-lookahead}. The optimal non-decreasing index policy is a direct corollary of \Cref{thm:opt-lookahead}, which is stated in \Cref{cor:sec3 cor}.

We start by introducing a definition to simplify our objective.
\begin{definition}[$\obj$]
    Let $\bar c_1(t)=\int_0^tc_1(s)\ ds$ refer to the cumulative holding cost incurred by a class 1 job when it has reached age $t$. We define $\obj:=\E[\bar c_1(T_1^\alpha)] - c_2 \frac{\mu_2}{ \mu_1}\E[T_1^\alpha]$.
\end{definition}

Now we apply the conservation law to the objective. 
\begin{lemma}
\label{lemma:obj}
Minimizing time-average total holding cost is equivalent to minimizing $\obj$.
\end{lemma}
\begin{proof}
    Since preemption is allowed, any reasonable policy is work conserving. Thus, we can apply the conservation law~\cite{kleinrock1965conservation}:
    \begin{equation}
        \rho_1 \E[T_1^\alpha] + \rho_2\E[T_2^\alpha] = W,
        \label{eq: conservation}
    \end{equation}
    where $W$ is the time average amount of work in an M/G/1 system with both class 1 jobs and class 2 jobs, which can be expressed as
    \[W=\frac{1}{1-\rho}\left(\frac{\lambda_1}{\mu_1^2} + \frac{\lambda_2}{\mu_2^2}\right).\]
    From \eqref{eq: conservation}, we have that \[\E[T_2^\alpha] = \frac{W}{\rho_2} - \frac{\rho_1}{\rho_2}\E[T_1^\alpha].\]

    This means the time-average total holding cost is equal to 
    \begin{align*}
        \E[Cost] &= \lambda_1 \E[\bar{c}_1(T_1^\alpha)] + \lambda_2 c_2 \E[T_2^\alpha] \\
        &= \lambda_1 \E[\bar{c}_1(T_1^\alpha)] + \lambda_2 c_2 \left(\frac{W}{\rho_2} - \frac{\rho_1}{\rho_2}\E[T_1^\alpha]\right) \\
        &= \lambda_1 \obj + \lambda_2 c_2 \frac{W}{\rho_2}.
    \end{align*}

    This means minimizing the time-average total holding cost is equivalent to minimizing $\obj.$
\end{proof}

We give a basic formula for exponential random variables.

\begin{lemma}
\label{lemma:exp formula}
    For any smooth function $f$ and exponential variable $X$, we have that 
    \[\E[f(x_0+X)]-f(x_0)=\E[X]\E[f'(x_0+X)].\]
    \[f(x_0)-\E[f(x_0-X)]=\E[X]\E[f'(x_0-X)].\]
\end{lemma}
\begin{proof}
    Suppose $X\sim\Exp(\theta)$. Then we have that 
    \begin{align*}
        \E[f(x_0+X)] -f(x_0) &= \int_{0}^\infty f'(x_0+t)\P{X>t}dt \\
        &=\int_0^\infty f'(x_0+t)e^{-\theta t} dt \\
        &= \frac{1}{\theta} \int_{0}^\infty f'(x_0+t) \theta e^{-\theta t} dt\\ 
        &= \E[X]\E[f'(x_0+X)].
    \end{align*}
    The other equation is derived similarly.
\end{proof}

\begin{lemma}\label{lem:lookahead-ECT}
    For any differentiable function $c(\cdot)$,
    \[
        \frac{d\E[c(T_1^\alpha)]}{d\alpha} = \E[c'(\alpha+ \Exp(\mu_1-\lambda_1))] \frac{d\E[T_1^\alpha]}{d\alpha}.
    \]
\end{lemma}
\begin{proof}
    We combine \Cref{lemma:exp formula,lemma: tail prob}. 
    \begin{align*}
        \E[c(T_1^\alpha)] &= c(0) + \int_0^\infty c'(t)\Pr[T_1^\alpha > t]dt\\
        &= c(0)+\int_0^\alpha  c'(t)\Pr[T_1^\alpha > t]dt + \int_\alpha^\infty  c'(t)\Pr[T_1^\alpha > t]dt\\
        &= c(0)+ \int_0^\alpha c'(t)\Pr[T_1^{\PPrio(2:1)}>t] dt \\
        &\qquad+\Pr[T_1^{\PPrio(2:1)}>\alpha] \int_0^\infty c'(\alpha+s)e^{-(\mu_1-\lambda_1)s} ds\tag{by \Cref{lemma: tail prob}}\\ 
        &=c(0)+ \int_0^\alpha c'(t)\Pr[T_1^{\PPrio(2:1)}>t] dt \\
        &\qquad+\Pr[T_1^{\PPrio(2:1)}>\alpha]\frac{1}{\mu_1-\lambda_1} \E[c'(\alpha+\Exp(\mu_1-\lambda_1))]\\
        &=c(0)+ \int_0^\alpha c'(t)\Pr[T_1^{\PPrio(2:1)}>t] dt \\
        &\qquad+\Pr[T_1^{\PPrio(2:1)}>\alpha]\left(\E[c(\alpha+\Exp(\mu_1-\lambda_1))] - c(\alpha)\right)\tag{by~\Cref{lemma:exp formula}}\\
        \text{Thus,}\\
        \frac{d\E[c(T_1^\alpha)]}{d\alpha} &=c'(\alpha)\Pr[T_1^{\PPrio(2:1)}>\alpha] + \frac{d\Pr[T_1^{\PPrio(2:1)} > \alpha]}{d\alpha}\parens*{ \E[c(\alpha+\Exp(\mu_1-\lambda_1))]-c(\alpha)}\\
        &\qquad + \Pr[T_1^{\PPrio(2:1)}>\alpha]\left(\E[c'(\alpha+\Exp(\mu_1-\lambda_1))] - c'(\alpha)\right)\tag{Dominated Convergence Theorem}\\
        &=
        \frac{d\Pr[T_1^{\PPrio(2:1)} > \alpha]}{d\alpha}\parens*{ \E[c(\alpha+\Exp(\mu_1-\lambda_1))]-c(\alpha)}\\
        &\qquad + \Pr[T_1^{\PPrio(2:1)}>\alpha]\E[c'(\alpha+\Exp(\mu_1-\lambda_1))] \\
      &= \parens*{\Pr[T_1^{\PPrio(2:1)}>\alpha]+\frac{d\Pr[T_1^{\PPrio(2:1)}>\alpha]}{d\alpha} \frac{1}{\mu_1-\lambda_1}}\E[c'(\alpha+\Exp(\mu_1-\lambda_1))]. \tag{by~\Cref{lemma:exp formula}}
    \end{align*}

    Note that this equation holds for any function $c$. By substituting the function $c(t)=t$, we have that
    \[\frac{d\E[T_1^\alpha]}{d\alpha}=\Pr[T_1^{\PPrio(2:1)}>\alpha]+\frac{d\Pr[T_1^{\PPrio(2:1)}>\alpha]}{d\alpha} \frac{1}{\mu_1-\lambda_1}.\]

Thus, we have that
    \[
         \frac{d\E[c(T_1^\alpha)]}{d\alpha}=\frac{d\E[T_1^\alpha]}{d\alpha}\cdot \E[c'(\alpha+\Exp(\mu_1-\lambda_1))].
    \]

\end{proof}

Finally, we characterize the optimal $\Overtake$ policy.

 \begin{theorem}[Optimal $\Overtake$ policy]\label{thm:opt-lookahead} The optimal $\Overtake$ policy, which we call $\Overtake(\alpha^*)$, is characterized as follows:
 
    \begin{itemize}
    \item If $\E[c_1(\Exp(\mu_1-\lambda_1))]>\frac{c_2\mu_2}{\mu_1}$, the optimal $\alpha^*$ is 0, which means that giving full priority to class 1 jobs is optimal.
    \item If $\lim_{t\to \infty}c_1(t)<\frac{c_2\mu_2}{\mu_1}$, the optimal $\alpha^*$ goes to infinity, which means that giving full priority to class 2 jobs is optimal.
    \item Otherwise, the optimal $\alpha^*$ satisfies 
    \begin{align}\label{eq:alpha-star}
        \E[c_1(\alpha^*+\Exp(\mu_1-\lambda_1))] = \frac{c_2\mu_2}{\mu_1}.
    \end{align}
    Further, the time-average total holding cost is convex in overtake age.
\end{itemize}
 \end{theorem}
\begin{proof}
Define $c(t)=\mu_1\bar{c}_1(t)-\mu_2c_2t$. Thus we have that for policy $\Overtake(\alpha)$, $\obj=\E[c(T_1^\alpha)]$.
    By \Cref{lem:lookahead-ECT}, we know
    \[
        \frac{d\obj}{d\alpha} = \parens*{\mu_1\E[c_1(\alpha+\Exp(\mu_1-\lambda_1))]-\mu_2c_2} \frac{d\E[T_1^\alpha]}{d\alpha}.
    \]
    Since increasing $\alpha$ makes class $1$ response times strictly larger, we have that $\frac{d\E[T_1^\alpha]}{d\alpha}>0$. Moreover, since $c_1$ is non-decreasing, $\mu_1\E[c_1(\alpha+\Exp(\mu_1-\lambda_1))]-\mu_2c_2$ is non-decreasing in $\alpha$. 

    Now if $\frac{d\obj}{d\alpha}>0$ for all $\alpha \geq 0$, then $\alpha^*=0$ is optimal. Likewise, if $\frac{d\obj}{d\alpha}<0$ for all $\alpha \geq 0$, the optimal $\alpha$ goes to infinity, which means full priority to class 2 jobs is optimal. 
    Otherwise $\frac{d\obj}{d\alpha}$ crosses $0$ exactly once at $\alpha^*$ defined in \Cref{eq:alpha-star}. This means the policy $\Overtake(\alpha^*)$ is the optimal $\Overtake$ policy.
\end{proof}
Finally, the optimal non-decreasing index policy follows directly from Assumption~\ref{assumption: class2 constant} and \Cref{thm:opt-lookahead}. In particular, we call this policy \LookAhead\ policy.
\begin{corollary}
\label{cor:sec3 cor}
    Given Assumption~\ref{assumption: class2 constant}, the optimal policy, called \LookAhead, among the class of non-decreasing index policies is given by the following index functions:
    \[
        V_1(t) = \mu_1\E[c_1(t+X)],  \qquad \text{where }X\sim\Exp(\mu_1-\lambda_1), \qquad \text{ and }V_2(t) = \mu_2c_2.
    \]
\end{corollary}
\begin{proof}
    Given Assumption~\ref{assumption: class2 constant}, the optimal non-decreasing index policy has a constant class 2 index. Therefore, it is equivalent to some $\Overtake$ policy. Moreover, the index policy defined by index functions $V_1(t) = \mu_1\E[c_1(t+X)]\text{ and }V_2(t) = \mu_2c_2$ is equivalent to $\Overtake(\alpha^*)$, the optimal $\Overtake$ policy. Thus we have the proof.
\end{proof}
   
\subsection{An intuitive interpretation of the \LookAhead\ policy via amortized cost}\label{sec:amortized}


Our policy stated in \Cref{cor:sec3 cor} can be seen as an \emph{amortized} $c\mu$ rule.

A class 1 job with age $t$ will accrue an expected remaining holding cost of $$\mbox{expected remaining holding cost } = \E[\bar{c}_1(T_1)|T_1\geq t]-\bar{c}_1(t).$$  It will accrue this cost over its expected remaining time in system, where
$$ \mbox{expected remaining time in system } = \E[T_1|T_1\geq t]-t. $$  
Suppose the job were to incur its expected remaining holding cost over its expected remaining time in system at a constant instantaneous holding cost rate, 
then the \emph{effective} expected cost rate (see \Cref{fig:amortized-cost}) would be
\[
    c_1^{\eff}(t) = \frac{\E[\bar{c}_1(T_1)|T_1\geq t]-\bar{c}_1(t)}{\E[T_1|T_1\geq t]-t}.
\]
\begin{figure}[htp]
    \centering
    \includegraphics[width=\linewidth]{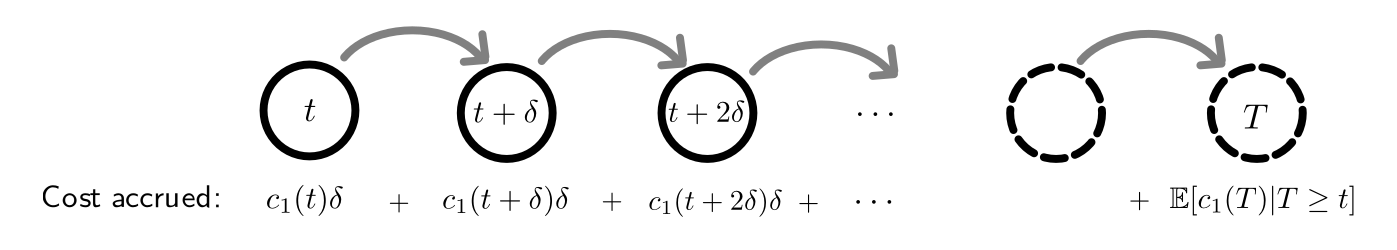}
    \caption{Expected Additional Cost accrued by a class 1 jobs which has reached age $t$.}
    \label{fig:amortized-cost}
\end{figure}
If we defined $c_2^{\eff}$ similarly, we would have $c_2^{\eff}(t)=c_2$ since $c_2(t)$ is constant. 
We could then apply the $c \mu$ rule to the {\em constant} holding costs $c_1^{\eff}$ and $c_2^{\eff}$ by always running the job with the highest $\mu_ic_i^{\eff}(t)$.  
Note that this policy definition is recursive since the response time distribution depends on the policy, and we are stipulating a policy which depends on the response time.

Interestingly, it turns out that our \LookAhead policy is an amortized $c\mu$ rule policy for the problem setting in our paper. To see this, note that since $c_2(t)$ is constant, we will always have $c_2^{\eff}(t)=c_2$. Further, if $c_1(\cdot)$ is increasing,
\[
    c_1^{\eff}(t) = \frac{\E[\int_t^{T_1|T_1\geq t}c_1(s)\ ds]}{\E[\int_t^{T_1|T_1\geq t}1\ ds]}
\]
is also increasing. As discussed in \Cref{sec:lookahead} and \Cref{fig:Lookahead}, this means the policy reduces to an $\Overtake(\alpha)$ policy for $\alpha$ such that 
\(
    \mu_1c_1^{\eff}(\alpha) = \mu_2c_2. 
\)
Then as discussed in \Cref{lemma:Q0}, class $1$ jobs enter an M/M/1 queue past their overtake age and experience response time
$[T_1|T_1\geq\alpha]=\alpha+X$, where $X\sim\Exp(\mu_1-\lambda_1)$. Thus we have
\[
    \mu_1c_1^{\eff}(\alpha) = \mu_1\cdot\frac{\E[\bar{c}_1(\alpha+X)-\bar{c}_1(\alpha)]}{\E[X]}=\mu_1\E[c_1(\alpha+X)]\tag{by \Cref{lemma:exp formula}}
\]
Therefore the amortized $c\mu$ rule's Overtake age $\alpha$ is exactly the optimal Overtake age $\alpha^*$ derived in \Cref{thm:opt-lookahead}. Equivalently, we can say that our \LookAhead policy always serves the job with highest amortized $c\mu$ index. This view gives us intuition for \emph{why} the optimal \LookAhead amount is $X$, where $X$ is the response time of an M/M/1 queue with only class 1 jobs. 


\section{A Longer Proof Without Assumption~\ref{assumption: class2 constant} }\label{sec:longer-proof}

In \Cref{sec:shorter-proof}, we proved that our policy is optimal under \Cref{assumption: class2 constant}. That is, our policy is optimal among those policies which have a non-decreasing index function for class $1$ and a constant index for class $2$. In this section, we relax \Cref{assumption: class2 constant}. Namely, we show that our policy \LookAhead (which is equivalently $\Overtake(\alpha^*)$) is optimal among policies where both index functions, $V_1$ and $V_2$, are non-decreasing.
Recall that $\Overtake(\alpha^*)$ is defined in Theorem~\ref{thm:opt-lookahead}. Specifically, in this section we only consider the interesting case where $\alpha^*$ satisfies \eqref{eq:alpha-star}. The proof of optimality for the case when $\alpha^*=0$ is a straightforward modification of our proof. For the case when $\alpha^*\to\infty$, an interchange argument can be used to prove that $\PPrio$(2;1) is optimal.

\subsection{Translating to R-MAB}

We begin by translating our holding cost problem to an R-MAB problem provided by \cite[Theorem 1]{li2025improving}. As stated in Lemma~\ref{lemma:obj}, our holding cost minimization problem is equivalent to a holding cost minimization problem where the holding cost of class 1 is $c(t)= \mu_1 c_1(t)- c_2\mu_2$, and class 2 jobs have zero holding cost. Applying \cite[Theorem 1]{li2025improving} to the new holding costs yields the following bandit.
\begin{theorem}[Theorem 1, \cite{li2025improving}]\label{thm:bandit}
For any set of non-decreasing index functions $\{V_i(\cdot)\}_{i=1}^k$, the corresponding index policies (breaking ties by FCFS) in the TVHC problem incur the same cost as in the following R-MAB problem:
\end{theorem}

\begin{itemize}
    \item The R-MAB has 2 arms, each representing a class. There are two actions for each arm (active or passive), where arm $i$ is active means the oldest class $i$ job is served, and passive means the job served at this moment is of the other type.   
    \item {\em States}: At any time, we can represent the state of the R-MAB system by the age of the oldest job of each class. Namely, we can represent the system state space as $\calS = \R^2$, where $(t_1, t_2)$ represents the state where the oldest class $i$ job has age $t_i$ for $i\in\{1,2\}$.  If there is no type $i$ job in the system, $t_i$ is negative, which means the next type $i$ arrival happens $(-t_i)$ time later. Note that for arm $i$, the action active is only allowed when $t_i$ is positive. 
    \item {\em Transition Probability: } If the action for arm $i$ is passive, the $i$th arm state grows with rate 1. Otherwise if arm $i$ is active, the $i$th arm state may drop an $\Exp(\lambda_i)$ amount according to a Poisson process (when completions happen). Mathematically, the transition function is 
    \[\text{passive: }dt_i = dt, \quad \text{active: }dt_i = dt - I \cdot dN_{\mu_i}(t),  \]
    $\text{where } I\sim \Exp(\lambda_i), \ \text{and} \ N_{\mu_i}  \text{ is a Poisson counting process with rate $\mu_i$.}$

    This transition function can be interpreted in the TVHC problem as follows: A passive action means that the oldest class $i$ job is not in service and its age grows with rate 1. If the action is active, then in the next $dt$ time period, the oldest class $i$ job is served. There is a probability of $\mu_i dt$ that the job is completed and leaves the system, in which case the oldest class $i$ job in the system becomes the previously second-oldest class $i$ job. Since the inter-arrival time follows the distribution $\Exp(\lambda_i)$, the age of the oldest job drops by an $\Exp(\lambda_i)$ amount.

    \item {\em Cost Function $r(t_1)$: } 
    for any system state $(t_1,t_2)$, define the cost incurred in the system to be 
    \begin{equation}
        r(t_1):= c(t_1) + \E[\sum_{j=1}^{N} c(Y_j)], 
        \label{eq:r}
    \end{equation}
    where $Y_j=\sum_{m=1}^j I_m$, $I_m \sim \Exp(\lambda_1)$, and $N$  is the random variable denoting the smallest stopping time such 
 that $Y_{N+1}\geq t_1$.   Observe that the second term of (\ref{eq:r}) represents the expected holding cost of all class 1 jobs present, except for the oldest class 1 job.  

Note that the cost function $r(t_1)$ depends only on the class $1$ state, $t_1$. This is consistent with our reformulation of the problem where class 1 has holding cost $c(t)$ and class 2 has zero holding cost.  
      
    \item {\em Objective: } The objective is to minimize the long-run expected cost. 
\end{itemize}


\subsection{Bellman Criteria for Optimality}



Note that the length of a busy period of a queueing system is identical for any work conserving policy. Thus, since our objective is to minimize the long-run cost incurred in a renewal system, by the Renewal Reward Theorem, it suffices to minimize the cost incurred in each busy period. 

A policy $\pi$ specifies which arm to pull at each state. Specifically, if $\pi$ pulls arm 1 at state $(t_1,t_2)$, we define $\pi(t_1,t_2)=1$.
Let $V_\pi(t_1, t_2)$ be the cost incurred starting from state $(t_1, t_2)$, until the end of the first busy period (i.e., the first time $t_1, t_2 < 0$), under policy $\pi$.  
Define the {\em Q-value} of a policy $\pi$, an action $a$ and a state $s$ to be the total cost incurred until the end of the busy period, given that we start in state $s$, do action $a$ for the next $\delta$ time, and act according to policy $\pi$ thereafter, until the end of the busy period.
Then the Q-values of policy $\pi$ for pulling arm $1$ or $2$ in state $(t_1, t_2)$ can be expressed as follows:
\begin{align}
    Q((t_1, t_2), \angles*{1, \delta, \pi}) &= r(t_1)\delta + (1-\mu_1\delta)V_\pi(t_1+\delta, t_2+\delta) \nonumber\\ &\quad + \mu_1\delta \E[V_\pi(t_1+\delta-I_1, t_2+\delta)] +o(\delta),\nonumber\\
    Q((t_1, t_2), \angles*{2, \delta, \pi}) &= r(t_1)\delta + (1-\mu_2\delta)V_\pi(t_1+\delta, t_2+\delta) \nonumber\\ &\quad + \mu_2\delta \E[V_\pi(t_1+\delta, t_2+\delta-I_2)]+o(\delta),
    \label{eq:Q values}
\end{align}
where $I_i\sim\Exp(\lambda_i)$ and $\angles{i, \delta, \pi}$ refers to the policy which pulls arm $i$ for the first $\delta$ time, then follows policy $\pi$. 

By Bellman optimality, a policy $\pi^*$ is optimal if and only if the following Bellman criteria is satisfied 
for any $t_1, t_2\geq 0$~\cite{Baird1994ReinforcementLI}:
\begin{align*}
    \text{if}\quad\pi^*(t_1, t_2) = 1  &,\quad\text{then}\quad \lim_{\delta\to 0}\frac{Q((t_1, t_2), \angles*{1, \pi^*}) - Q((t_1, t_2), \angles*{2, \pi^*})}{\delta}\leq 0, \\
    \text{if}\quad\pi^*(t_1, t_2) = 2 &,\quad\text{then}\quad\lim_{\delta\to 0}\frac{Q((t_1, t_2), \angles*{1, \pi^*}) - Q((t_1, t_2), \angles*{2, \pi^*})}{\delta}\geq 0. \\
\end{align*}

By \eqref{eq:Q values}, this is equivalent to 
\begin{align}
    \text{if}\quad\pi^*(t_1, t_2) = 1  &,\ \text{then}\quad \nonumber\\
    \mu_1(V_{\pi^*}&(t_1, t_2) - E[V_{\pi^*}(t_1-I_1, t_2)]) \geq \mu_2(V_{\pi^*}(t_1, t_2) - E[V_{\pi^*}(t_1, t_2-I_2)]),\nonumber\\
    \text{if}\quad\pi^*(t_1, t_2) = 2 &,\ \text{then}\quad \nonumber\\
    \mu_2(V_{\pi^*}&(t_1, t_2) - E[V_{\pi^*}(t_1-I_1, t_2)]) \leq \mu_2(V_{\pi^*}(t_1, t_2) - E[V_{\pi^*}(t_1, t_2-I_2)]).\label{eq:Bellman}
\end{align}

In the following subsections, we will show that our policy $\Overtake(\alpha^*)$ satisfies the above Bellman optimality criteria for the R-MAB problem. 

$\Overtake(\alpha^*)$ prefers arm $1$ over arm $2$ in state $(t_1, t_2)$ iff $t_1\geq \alpha^*$. For ease of notation, we omit expectations and for example write $V(t_1-I_1, t_2-I_2)$ to refer to the expectation  $\E_{I_1, I_2}\bracks*{V_{\Overtake(\alpha^*)}(t_1-I_1, t_2-I_2)}$. Under $\Overtake(\alpha^*)$, \eqref{eq:Bellman} transform into the following claims:
For all $t_2\geq 0$, 
\begin{description}
    \item[Claim 1.] If $t_1\geq \alpha^*$, $\mu_1(V(t_1, t_2) - V(t_1 - I_1, t_2)) \geq \mu_2(V(t_1, t_2)-V(t_1, t_2-I_2))$.
    \item[Claim 2.] If $0\leq t_1<\alpha^*$, $\mu_1(V(t_1, t_2) - V(t_1 - I_1, t_2)) \leq \mu_2(V(t_1, t_2)-V(t_1, t_2-I_2))$. 
     
\end{description}

The following definition provides notation for a busy period started by some initial work.  We will use this notation throughout the rest of this section.
\begin{definition}[Busy period with initial work]
    We denote the length of a busy period started by $W_0$ initial work, and with Poisson arrivals of rate $\lambda$ and job-size distribution $S$, as $\BP[W_0; \lambda, S]$. 
\end{definition}

We also provide notation characterizing the 2-class M/M/1 queue, which we will use throughout the rest of this section.

\begin{definition}[$\lambda, S$]
    Let the total mean arrival rate to the two class M/M/1 queue be denoted as 
    \[
        \lambda = \lambda_1 + \lambda_2.
    \]
    Let the job size distribution of class $i$ be $S_i\sim\Exp(\mu_i)$ for $i\in\{1,2\}$. The overall job size distribution is denoted
    \[
        S =
        \begin{cases}
            S_1, &\text{with probability }\frac{\lambda_1}{\lambda},\\
            S_2, &\text{with probability }\frac{\lambda_2}{\lambda}.
        \end{cases}
    \]
\end{definition}

\subsection{Verifying Claim 1}
The goal of this section is to verify Claim 1 (Proposition~\ref{claim 1}). We start by defining the notation for the total class 1 work arriving during a $t$ period of time.

\begin{definition}
    Define $\wideW(t)$ to be the total class 1 work arriving during a $t$ period of time. Mathematically, \[\wideW(t):=\sum_{i=1}^{N}S_1^{(i)},\text{where } S_1^{(i)}\sim \Exp(\mu_1),\ N\sim \Pois(\lambda_1 t).\]
\end{definition}

We define a random process $A$. Intuitively, $A$ represents the dynamics of the oldest age in an M/M/1 queue with arrival rate $\lambda_1$ and job size distribution $\Exp(\mu_1)$.
\begin{definition}[$A$]
    Given an initial state $A(0)$, we define the following random process $A$: At each time $t$,
    \begin{itemize}
        \item if $A(t)<0,$ then $dA(t) = dt$;
        \item if $A(t)\geq 0$, then $dA(t) = dt - I \cdot dN_{\mu_1}(t),$ where $I\sim \Exp(\lambda_1)$ and $N_{\mu_1}$ is a Poisson counting process with rate $\mu_1$.
    \end{itemize}
\end{definition}

Now we introduce a lemma characterizing the value function of $\Overtake(\alpha^*)$.
\begin{lemma}
\label{lemma: V claim 1}
For any $t_1\geq \alpha^*$,
    \begin{align}
        V(t_1, t_2) &= \E\bracks*{\int_0^{\BP[S_1+S_2+\wideW(t_1-\alpha^*); \lambda, S]}r(\alpha^*+A(t))\ dt \Big| A(0)=t_1-\alpha^*}   \nonumber\\ &\quad + V(\alpha^*-I_1, t_2 - I_2)\\
        V(t_1, t_2-I_2) &= \E\bracks*{\int_0^{\BP[S_1+\wideW(t_1-\alpha^*); \lambda, S]}r(\alpha^*+A(t))\ dt \Big| A(0)=t_1-\alpha^*}  \nonumber\\ &\quad + V(\alpha^*-I_1, t_2 - I_2)\\
        V(t_1-I_1, t_2) &= \E\bracks*{\int_0^{\BP[S_2+\wideW(t_1-\alpha^*); \lambda, S]}r(\alpha^*+A(t))\ dt \Big| A(0)=t_1-\alpha^*-I_1} \nonumber\\ &\quad + V(\alpha^*-I_1, t_2 - I_2)
    \end{align}
\end{lemma}

\begin{proof}
    We view the system through ``glasses'' where we only see class $1$ jobs of age $\geq \alpha^*$ and class $2$ jobs of age $\geq t_2$, as in \Cref{fig:glasses}. We call these jobs \emph{active}. 
    \begin{figure}
        \centering
        \includegraphics[width=\linewidth]{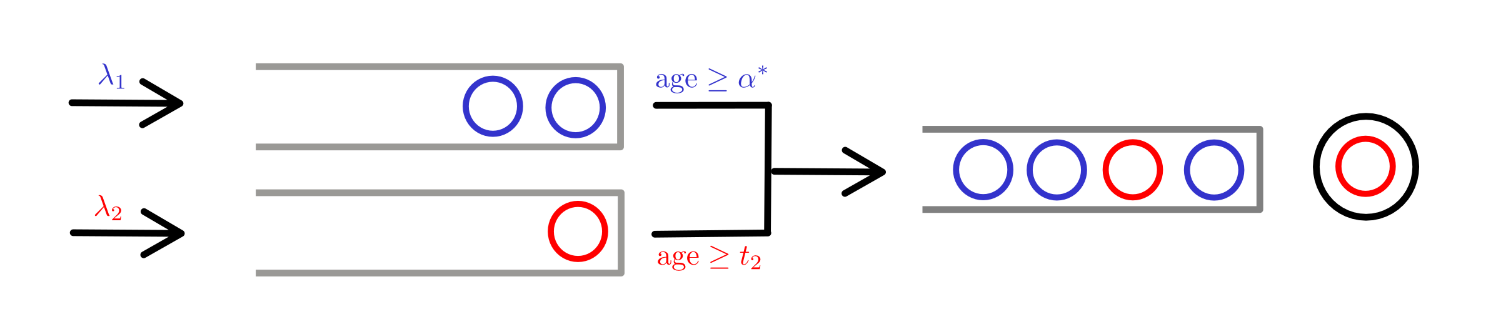}
        \caption{Only jobs in the active queue (on the right) are visible to the server.}
        \label{fig:glasses}
    \end{figure}
    
    In state $(t_1, t_2)$, the system starts with one visible job of each class, and Poisson arrivals behind these oldest jobs. Under the policy $\Overtake(\alpha^*)$, we will serve only active jobs as long as we see them, and we will serve them $\PPrio(1;2)$ for this time. Further, while we are only serving active jobs, we will see Poisson arrivals of both classes to the active queue. 
    
    Since the initial active work is the oldest job of each class (of age respectively $t_1$ and $t_2$), and also any younger class 1 jobs with age larger than $\alpha^*$, the total initial active work is $S_1+S_2+\wideW(t_1-\alpha^*)$.
    Thus we will serve active jobs for a total busy period $\BP[S_1+S_2+\wideW(t_1-\alpha^*);\lambda, S]$. The busy period ends at a departure, when the next oldest job of each class is younger than $\alpha^*$ and $t_2$ respectively: the ages of the oldest jobs in class 1 and class 2 are respectively distributed as $\alpha^*-I_1$ and $t_2-I_2$ at the end of the busy period. This means at the end of the busy period, the system state is distributed as $(\alpha^*-I_1, t_2-I_2)$.

    Similarly, in states $(t_1, t_2-I_2)$ and $(t_1-I_1, t_2)$, we serve jobs for busy periods $\BP[S_1+\wideW(t_1-\alpha^*);\lambda, S]$ and $\BP[S_2+\wideW(t_1-\alpha^*);\lambda, S]$ respectively. At the end of the busy period, we reach the same state distribution $(\alpha^*-I_1, t_2-I_2)$. 
\end{proof}

 Thus, using this lemma, to prove Proposition~\ref{claim 1} it suffices to characterize the expectation terms in Lemma~\ref{lemma: V claim 1}. The approach we adopt is: We strategically introduce a discount factor $\beta$ into the expectation terms, and use Lemmas from \cite{li2025improving} to characterize the discounted terms. Then we take the limit $\beta\to 0$ to obtain the expressions for the expectation terms in Lemma~\ref{lemma: V claim 1}.

 \subsubsection{Introducing the discount factor}
In order to characterize the expectation terms in Lemma~\ref{lemma: V claim 1}, we introduce a discount factor $\beta$ and define the following discounted version of our expectations.
\begin{definition}
For any $t_1\geq \alpha^*$, define 
    \begin{align}
\Omega_\beta(t_1, W) &= \E\bracks*{\int_0^{\BP[W;\lambda; S]}r(\alpha^*+A(t))e^{-\beta t}\ dt \ \big|\ A(0) = t_1-\alpha^*}.         
    \end{align}
\end{definition}

Note that the expectation terms in Lemma~\ref{lemma: V claim 1} can be captured by $\Omega_\beta$ by taking the limit $\beta\to 0$. We use the following term as an intermediate step to get the expression for $\Omega_\beta$. 

\begin{definition}
For any $t_1\geq \alpha^*$, define 
    \begin{align}
        U_\beta (t_1) &=\E\bracks*{ \int_0^\infty r(\alpha^*+A(t))e^{-\beta t} \ dt\ |\ A(0) = t_1-\alpha^*}.
    \end{align}
\end{definition}


The following subsection manipulates several lemmas from \cite{li2025improving} to get an expression of $U_\beta$, and then in section~\ref{sec:sec claim 1} we return to $\Omega_\beta$ and finally the expectation terms in Lemma~\ref{lemma: V claim 1}.

\subsubsection{Characterizing $U_\beta$}

In this subsection, lemmas from \cite{li2025improving} are used to characterize $U_\beta$. 
We adapt the following definitions from \cite{li2025improving}. 

\begin{definition}[Definitions in \cite{li2025improving}]
\label{def:whittle}
For any variable $Y$, define $\tilde{Y}(x):=\E[e^{-xY}]$ to be its Laplace transform. Define 
\begin{align}
        \gamma_1 &:= \E[e^{-\beta \BP[S_1;\lambda_1; S_1]}]=\widetilde{\BP}[S_1;\lambda_1; S_1](\beta),\\
    \gamma_2 &:= \E[e^{-\beta I_1}]=\tilde{I}_1(\beta),\\
    X &\sim \Exp\parens*{\frac{\beta}{1-\gamma_1}}, \\
    \costbar(t,\beta)&:= \E[\int_0^{BP[S_1;\lambda_1;S_1]} r(t+A(t))\cdot \beta e^{-\beta t} dt\ \Big |\  A(0)=0],\\
    \Gamma(t,\beta)&:= \E[\int_0^{I_1} r(t-I_1+t) \cdot \beta e^{-\beta t} dt],\\
    \PoisCost(t,\alpha^*)&:= \E[\sum_{i=1}^{M} \gamma_1^{i-1}\costbar(y_i,\beta)],
    \end{align}

    where $y_1=t-\Exp(\lambda_1),y_{i+1}=y_i-\Exp(\lambda)$, and $M$ is the random variable such that $y_M\geq \alpha^*,y_{M+1}<\alpha^*$, with $M\sim \Pois(\lambda_1(t-\alpha^*))$.
\end{definition}

The following lemma from \cite{li2025improving} gives the closed form of these terms.
\begin{lemma}[Lemmas in \cite{li2025improving}] For the terms defined in Definition~\ref{def:whittle}, we have the following expressions:
    \begin{equation}
        \costbar(t,\beta) = (1-\gamma_1)\E[r(t+X)], \qquad \mbox{ where } X \sim \Exp\left(\frac{\beta}{1 - \gamma_1}\right),
    \end{equation}
    \begin{equation}
    \label{eq: equation2}
        \Gamma(t,\beta)= (1-\gamma_2)\E[r(t-I_1)],
    \end{equation}
    \begin{equation}
        \PoisCost(t,\alpha^*) = \int_{\alpha^*}^t e^{-\lambda_1(1-\gamma_1)(t-s)}\lambda_1 \costbar(s,\beta) ds,
         \label{eq:PoisCost}
    \end{equation}
    \begin{equation}
        \E[\gamma_1^{M}]= PGF(\gamma_1)= e^{-\lambda(t-\alpha^*)(1-\gamma_1)}.
        \label{eq:PGF}
    \end{equation}
\end{lemma}

\begin{proof}
    See Lemma 4.17, 4.18 and Case 2 of Lemma E.1 in \cite{li2025improving}.
\end{proof}

We can use the terms defined in Definition~\ref{def:whittle} to get the characterization of $U_\beta$.

\begin{lemma}
    \label{lemma:U formula from Whittle}
    For any $t_1\geq \alpha^*$,
    \begin{align}
    U_\beta(\alpha^*-I_1) &= \frac{1}{1-\gamma_1\gamma_2}\left(\frac{1}{\beta}\Gamma(\alpha^*,\beta) + \gamma_2 \frac{1}{\beta}\costbar(\alpha^*,\beta) \right) \label{eq:eq 22}\\
        U_\beta(t_1) &= \frac{1}{\beta}\left(\costbar(t,\beta) + \gamma_1 \PoisCost(t_1,\alpha^*)\right)+ \E[\gamma_1^{M+1}]U_\beta(\alpha^*-I_1). \label{eq:eq 23}\\
        U_\beta(t_1-I_1) &= \frac{1}{\beta}\PoisCost(t_1,\alpha^*)+ \E[\gamma_1^{M}]U_\beta(\alpha^*-I_1).\label{eq:eq 24}
    \end{align}
    
\end{lemma}

\begin{proof}
We start by proving \eqref{eq:eq 22}.
    Starting at state $\alpha^*-I_1$, there is first an $I_1$ period of time before $A(t)$ returns to 0. During this period of time, $\frac{1}{\beta}\Gamma(\alpha^*,\beta)$ cost is incurred. Then starting at $A=0$, by definition, a cost of $\frac{1}{\beta}\costbar(\alpha^*,\beta)$ is incurred during a $BP(S_1;\lambda_1;S_1)$ time period. Finally, $A$ returns to the state $\alpha^*-I_1$ (because each drop is $\Exp(\lambda_1)$).  Given that $\gamma_1,\gamma_2$ are respectively the discounting factor after time period $I_1$ and $BP(S_1;\lambda_1;S_1)$, we have that 
    \[U_\beta(\alpha^*-I_1) = \frac{1}{\beta}\Gamma(\alpha^*,\beta) + \gamma_2 \frac{1}{\beta}\costbar(\alpha^*,\beta) + \gamma_1\gamma_2 U_\beta(\alpha^*-I_1).\]
    This proves the first equation. 

    To prove \eqref{eq:eq 23}, we use a similar argument as in Lemma E.1 in \cite{li2025improving}: Starting from $A=t_1$, the policy stays active until the state drops below $t_1$. During this time ($BP(S_1;\lambda_1;S_1)$), a cost of $\costbar(t_1,\beta)$ is incurred. 
    After the state drops below $t_1$, the amount it is below $t_1$ follows an exponential distribution with rate $\lambda_1$. Suppose it is $y_1\sim t_1-\Exp(\lambda_1)$. If $y_1$ is still larger than $\alpha^*$, the policy stays active until the state drops below $y_1$, incurring another $\costbar(y_1,\beta)$ cost. This process continues until the state drops below $\alpha^*$. Note that each time a $\costbar(y_i,\beta)$ is incurred, the state drops by an exponential amount, thus the total number of iterations follow a Poisson distribution $\Pois(\lambda_1(t_1-\alpha^*))$. After the final iteration, the state is $\Exp(\lambda_1)$ below $\alpha^*$. 
    This process yields \eqref{eq:eq 23}.
    
  \Cref{eq:eq 24} can be proved by a similar argument.
\end{proof}

Finally, we take the limit $\beta\to 0$ on the expressions. 
We first give the following well-known lemma characterizing the limit of a Laplace transform.

\begin{lemma}
\label{lemma: limit on laplace}
    Given a random variable $Y$, we have that
    \[\lim_{\beta\to 0}\frac{1-\tilde{Y}(\beta)}{\beta}=\E[Y].\]
\end{lemma}
\begin{proof} The proof follows immediately from L'Hospital's Rule and the property of the Laplace transform:
    \[\lim_{\beta\to 0} \frac{1-\tilde{Y}(\beta)
        }{\beta} = -\tilde{Y}'(0) = \E[Y].\]
\end{proof}

Thus we can take the limit $\beta\to 0$ on the random variables.
\begin{lemma}
\label{lemma:Ualpha}
When $\beta\to 0$, we have that
    \begin{align}
        \lim_{\beta\to 0}\frac{\costbar(t,\beta)}{\beta} &=  \frac{1}{\mu_1-\lambda_1}\E[r(t+\Exp(\mu_1-\lambda_1))], \\
        \lim_{\beta\to 0} \frac{\Gamma(t,\beta)}{\beta} &= \frac{1}{\lambda_1}\E[r(t-I_1)], \\
        \lim_{\beta\to 0}\frac{
        \PoisCost(t,\alpha^*)}{\beta} &= \frac{\lambda_1}{\mu_1-\lambda_1}\int_{\alpha^*}^t \E[r(s+\Exp(\mu_1-\lambda_1))]ds, \\
        \lim_{\beta\to 0}\frac{1-\E[\gamma_1^{M}]}{\beta}&= \lambda_1 (t-\alpha^*)\frac{1}{\mu_1-\lambda_1},\\
        \lim_{\beta\to 0} \beta U_\beta(\alpha^*-I_1)&= (\mu_1-\lambda_1)\E[r(\alpha^*-I_1)]+\lambda_1\E[r(\alpha^*+\Exp(\mu_1-\lambda_1))].
    \end{align}
\end{lemma}
\begin{proof}
    \begin{align}
        \lim_{\beta\to 0}\frac{\costbar(t,\beta)}{\beta} &= \lim_{\beta\to 0} \frac{1-\gamma_1}{\beta} \E[r(t+\Exp(\frac{\beta}{1-\gamma_1}))] \nonumber\\
        &=\E[BP[\lambda_1;S_1]] \cdot \E[r(t+\Exp(\frac{1}{\E[BP[\lambda_1;S_1]]}))]\tag{Dominated Convergence Theorem}\nonumber\\
        &= \frac{1}{\mu_1-\lambda_1}\E[r(t+\Exp(\mu_1-\lambda_1))].
    \end{align}

    \begin{align}
        \lim_{\beta\to 0} \frac{\Gamma(t,\beta)}{\beta} &=  \lim_{\beta\to 0} \frac{1-\gamma_2}{\beta} \E[r(t-I_1)] \nonumber\\
        &= \frac{1}{\lambda_1}\E[r(t-I_1)].
    \end{align}

    \begin{align}
    \lim_{\beta\to 0}\frac{
        \PoisCost(t,\alpha^*)}{\beta} &= \int_{\alpha^*}^t \lim_{\beta\to 0} e^{-\lambda_1(1-\gamma_1)(t-s)}\lambda_1 \frac{1}{\beta}\costbar(s,\beta) ds \nonumber\\
        &= \int_{\alpha^*}^t \lim_{\beta\to 0}\lambda_1\frac{1}{\beta}\costbar(s,\beta) ds \nonumber &&\text{because $\gamma_1\to 1$}\\
        &= \frac{\lambda_1}{\mu_1-\lambda_1}\int_{\alpha^*}^t \E[r(s+\Exp(\mu_1-\lambda_1))]ds.
    \end{align}

    \begin{align}
        \lim_{\beta\to 0}\frac{1-\E[\gamma_1^{M}]}{\beta} &= \lim_{\beta\to 0} \frac{1-e^{-\lambda_1(t-\alpha^*)(1-\gamma_1)}}{\beta} \nonumber\\
        &= \lambda_1 (t-\alpha^*)\lim_{\beta\to 0}\frac{1-\gamma_1}{\beta} && \text{L'Hospital's rule} \nonumber\\
        &=  \lambda_1 (t-\alpha^*)\frac{1}{\mu_1-\lambda_1}.
    \end{align}

    \begin{align}
        &\lim_{\beta\to 0} \beta U_\beta(\alpha^*-I_1) \nonumber\\
        =& \lim_{\beta\to 0} \frac{\beta}{1-\gamma_1\gamma_2}\left(\frac{1}{\beta}\Gamma(\alpha^*,\beta) + \gamma_2 \frac{1}{\beta}\costbar(\alpha^*,\beta) \right) \nonumber \\
        =& (\mu_1-\lambda_1)\lambda_1\left( \frac{1}{\lambda_1}\E[r(\alpha^*-I_1)] + \frac{1}{\mu_1-\lambda_1}\E[r(\alpha^*+\Exp(\mu_1-\lambda_1))] \right)\nonumber\\
        =& (\mu_1-\lambda_1)\E[r(\alpha^*-I_1)]+\lambda_1\E[r(\alpha^*+\Exp(\mu_1-\lambda_1))].
    \end{align}
    
\end{proof}

\subsubsection{Proving Claim 1}
\label{sec:sec claim 1}

Finally, we derive the expected terms in Lemma~\ref{lemma: V claim 1} by taking the limit $\beta\to 0$ in the $\Omega_\beta$ terms, which can be characterized using the $U_\beta$ terms.

\begin{lemma}
\label{lemma:expected terms}
The expected terms in Lemma~\ref{lemma: V claim 1} have the expressions given in \eqref{eq:eq35}, \eqref{eq:eq36} and \eqref{eq:eq37}.
\end{lemma}
\begin{proof}
    Define $B_1:=BP[S_1;\lambda;S]$, $B_2:=BP[S_2;\lambda;S]$, $B_W:=BP[\wideW(t_1-\alpha^*);\lambda;S]$. 

    Since $\Omega_\beta(t,W)$ can be seen as the total cost incurred during the first $BP[W;\lambda;S]$ time, we have that 
    \begin{align*}
        U_\beta(t_1) &= 
        \Omega_\beta(t_1, S_1+S_2+\wideW(t_1-\alpha^*))+\tilde{B}_1(\beta)\tilde{B}_2(\beta)\tilde{B}_W(\beta)U_\beta(\alpha^*-I_1),\\
        U_\beta(t_1) &= 
        \Omega_\beta(t_1, S_1+\wideW(t_1-\alpha^*))+ \tilde{B}_1(\beta)\tilde{B}_W(\beta)U_\beta(\alpha^*-I_1),\\
         U_\beta(t_1-I_1)  &= \Omega_\beta(t_1-I_1, S_2+\wideW(t_1-\alpha^*)) + \tilde{B}_2(\beta)\tilde{B}_W(\beta)U_\beta(\alpha^*-I_1),
    \end{align*}
    where $\tilde{B}(\beta)=\E[e^{-\beta B}]$ is the expected discount factor after a busy period $B$.

    Now we take the limit $\beta\to 0$ and use Lemma~\ref{lemma:U formula from Whittle} and Lemma~\ref{lemma:Ualpha}.

    \begin{align}
    &\E\bracks*{\int_0^{\BP[S_1+S_2+\wideW(t_1-\alpha^*); \lambda, S]}r(\alpha^*+A(t))\ dt \Big| A(0)=t_1-\alpha^*} \nonumber\\
    =& \lim_{\beta\to 0} \Omega_\beta(t_1, S_1+S_2+\wideW(t_1-\alpha^*)) \nonumber\\
    =& \lim_{\beta \to 0} U_\beta(t_1) - \tilde{B}_1(\beta)\tilde{B}_2(\beta)\tilde{B}_W(\beta)U_\beta(\alpha^*-I_1) \nonumber\\
    =& \lim_{\beta \to 0} \frac{1}{\beta}\left(\costbar(t_1,\beta) + \gamma_1 \PoisCost(t_1,\alpha^*)\right)+ \left(\E[\gamma_1^{M+1}] -\tilde{B}_1(\beta)\tilde{B}_2(\beta)\tilde{B}_W(\beta)\right)U_\beta(\alpha^*-I_1)  \nonumber\\
    =& \frac{1}{\mu_1-\lambda_1}\E[r(t_1+\Exp(\mu_1-\lambda_1))] + \frac{\lambda_1}{\mu_1-\lambda_1}\int_{\alpha^*}^{t_1} \E[r(s+\Exp(\mu_1-\lambda_1))]ds \nonumber\\
    &+ \lim_{\beta\to 0}\left(\frac{1-\tilde{B}_1(\beta)\tilde{B}_2(\beta)\tilde{B}_W(\beta)}{\beta} - \frac{1-\gamma_1 \E[\gamma_1^M]}{\beta}\right) \beta U_\beta(\alpha^*-I_1)\nonumber\\
    =& \frac{1}{\mu_1-\lambda_1}\E[r(t_1+\Exp(\mu_1-\lambda_1))] + \frac{\lambda_1}{\mu_1-\lambda_1}\int_{\alpha^*}^{t_1} \E[r(s+\Exp(\mu_1-\lambda_1))]ds \nonumber\\
    &+ \left(\E[B_1+B_2+B_W] -\lambda_1 (t_1-\alpha^*)\frac{1}{\mu_1-\lambda_1} - \frac{1}{\mu_1-\lambda_1}\right) \lim_{\beta\to 0}\beta U_\beta(\alpha^*-I_1).& Lemma~\ref{lemma: limit on laplace}
    \label{eq:eq35}
    \end{align}

Similarly, we have that 
\begin{align}
    &\E\bracks*{\int_0^{\BP[S_1+\wideW(t_1-\alpha^*); \lambda, S]}r(\alpha^*+A(t))\ dt \Big| A(0)=t_1-\alpha^*} \nonumber\\
    =& \lim_{\beta\to 0} \Omega_\beta(t_1, S_1+\wideW(t_1-\alpha^*)) \nonumber\\
    =& \lim_{\beta \to 0} U_\beta(t_1) - \tilde{B}_1(\beta)\tilde{B}_W(\beta)U_\beta(\alpha^*-I_1) \nonumber\\
    =& \lim_{\beta \to 0} \frac{1}{\beta}\left(\costbar(t_1,\beta) + \gamma_1 \PoisCost(t_1,\alpha^*)\right)+ \left(\E[\gamma_1^{M+1}] -\tilde{B}_1(\beta)\tilde{B}_W(\beta)\right)U_\beta(\alpha^*-I_1)  \nonumber\\
    =& \frac{1}{\mu_1-\lambda_1}\E[r(t_1+\Exp(\mu_1-\lambda_1))] + \frac{\lambda_1}{\mu_1-\lambda_1}\int_{\alpha^*}^{t_1} \E[r(s+\Exp(\mu_1-\lambda_1))]ds \nonumber\\
    &+ \lim_{\beta\to 0}\left(\frac{1-\tilde{B}_1(\beta)\tilde{B}_W(\beta)}{\beta} - \frac{1-\gamma_1 \E[\gamma_1^M]}{\beta}\right) \beta U_\beta(\alpha^*-I_1)\nonumber\\
    =& \frac{1}{\mu_1-\lambda_1}\E[r(t_1+\Exp(\mu_1-\lambda_1))] + \frac{\lambda_1}{\mu_1-\lambda_1}\int_{\alpha^*}^{t_1} \E[r(s+\Exp(\mu_1-\lambda_1))]ds \nonumber\\
    &+ \left(\E[B_1+B_W] -\lambda_1 (t_1-\alpha^*)\frac{1}{\mu_1-\lambda_1} - \frac{1}{\mu_1-\lambda_1}\right) \lim_{\beta\to 0}\beta U_\beta(\alpha^*-I_1).& Lemma~\ref{lemma: limit on laplace}\label{eq:eq36}
    \end{align}

\begin{align}
    &\E\bracks*{\int_0^{\BP[S_2+\wideW(t_1-\alpha^*); \lambda, S]}r(\alpha^*+A(t))\ dt \Big| A(0)=t_1-\alpha^*-I_1} \nonumber\\
    =& \lim_{\beta\to 0} \Omega_\beta(t_1-I_1, S_2+\wideW(t_1-\alpha^*)) \nonumber\\
    =& \lim_{\beta \to 0} U_\beta(t_1-I_1) - \tilde{B}_2(\beta)\tilde{B}_W(\beta)U_\beta(\alpha^*-I_1) \nonumber\\
    =& \lim_{\beta \to 0}\frac{1}{\beta}\PoisCost(t_1,\alpha^*)+\left(\E[\gamma_1^{M}] -\tilde{B}_2(\beta)\tilde{B}_W(\beta)\right)U_\beta(\alpha^*-I_1)  \nonumber\\
    =& \frac{\lambda_1}{\mu_1-\lambda_1}\int_{\alpha^*}^{t_1} \E[r(s+\Exp(\mu_1-\lambda_1))]ds \nonumber\\
    &+ \lim_{\beta\to 0}\left(\frac{1-\tilde{B}_2(\beta)\tilde{B}_W(\beta)}{\beta} - \frac{1- \E[\gamma_1^M]}{\beta}\right) \beta U_\beta(\alpha^*-I_1)\nonumber\\
    =& \frac{\lambda_1}{\mu_1-\lambda_1}\int_{\alpha^*}^{t_1} \E[r(s+\Exp(\mu_1-\lambda_1))]ds \nonumber\\
    &+ \left(\E[B_2+B_W] -\lambda_1 (t_1-\alpha^*)\frac{1}{\mu_1-\lambda_1}\right) \lim_{\beta\to 0}\beta U_\beta(\alpha^*-I_1).& Lemma~\ref{lemma: limit on laplace}\label{eq:eq37}
    \end{align}
    
\end{proof}

We list here one more lemma in \cite{li2025improving} that will be useful.
\begin{lemma}[$r'$]
\label{lemma: r'}
The derivative of $r$ is given by
    \[r'(t) = c'(t) + \lambda_1 c(t).\]
\end{lemma}
\begin{proof}
    See Lemma 4.16 in \cite{li2025improving}.
\end{proof}

Now we are ready to prove Claim 1.
\begin{proposition}[Claim 1]
    For all $t_1\geq \alpha^*, t_2\geq 0$, 
    \[
        \mu_1(V(t_1, t_2) - V(t_1 - I_1, t_2)) \geq \mu_2(V(t_1, t_2)-V(t_1, t_2-I_2)).
    \]
    Furthermore, the inequality is an equality when $t_1=\alpha^*$.
    \label{claim 1}
\end{proposition}

\begin{proof}
 By Lemma~\ref{lemma: V claim 1}, it suffices to show that 
 \begin{align}
     &\mu_1\Big(\E\bracks*{\int_0^{\BP[S_1+S_2+\wideW(t_1-\alpha^*); \lambda, S]}r(\alpha^*+A(t))\ dt \Big| A(0)=t_1-\alpha^*}\nonumber\\ &- \E\bracks*{\int_0^{\BP[S_2+\wideW(t_1-\alpha^*); \lambda, S]}r(\alpha^*+A(t))\ dt \Big| A(0)=t_1-\alpha^*-I_1}\Big) \nonumber\\
     &\geq \mu_2\Big(\E\bracks*{\int_0^{\BP[S_1+S_2+\wideW(t_1-\alpha^*); \lambda, S]}r(\alpha^*+A(t))\ dt \Big| A(0)=t_1-\alpha^*}\nonumber\\ &- \E\bracks*{\int_0^{\BP[S_1+\wideW(t_1-\alpha^*); \lambda, S]}r(\alpha^*+A(t))\ dt \Big| A(0)=t_1-\alpha^*}\Big).
     \label{eq:eq1}
 \end{align}

By Lemma~\ref{lemma:expected terms}, we have that the left hand side is equal to 
\begin{align}
    &\mu_1\Big(\E\bracks*{\int_0^{\BP[S_1+S_2+\wideW(t_1-\alpha^*); \lambda, S]}r(\alpha^*+A(t))\ dt \Big| A(0)=t_1-\alpha^*}\nonumber\\ &- \E\bracks*{\int_0^{\BP[S_2+\wideW(t_1-\alpha^*); \lambda, S]}r(\alpha^*+A(t))\ dt \Big| A(0)=t_1-\alpha^*-I_1}\Big) \nonumber\\
    &= \mu_1\left(\frac{1}{\mu_1-\lambda_1}\E[r(t_1+\Exp(\mu_1-\lambda_1))] + \left(\E[B_1] - \frac{1}{\mu_1-\lambda_1}\right)\lim_{\beta\to 0}\beta U_\beta(\alpha^*-I_1)\right).
\end{align}

Also, the right hand side of \eqref{eq:eq1} is equal to:
\begin{align}
    &\mu_1\Big(\E\bracks*{\int_0^{\BP[S_1+S_2+\wideW(t_1-\alpha^*); \lambda, S]}r(\alpha^*+A(t))\ dt \Big| A(0)=t_1-\alpha^*}\nonumber\\ &- \E\bracks*{\int_0^{\BP[S_1+\wideW(t_1-\alpha^*); \lambda, S]}r(\alpha^*+A(t))\ dt \Big| A(0)=t_1-\alpha^*}\Big) \nonumber\\
    &= \mu_2\E[B_2]\lim_{\beta\to 0}\beta U_\beta(\alpha^*-I_1).
\end{align}

Note that $\E[B_1]=\frac{\E[S_1]}{1-\rho}=\frac{1}{\mu_1(1-\rho)}$, $\E[B_2]=\frac{1}{\mu_2(1-\rho)}.$ Thus we have that \eqref{eq:eq1} is equivalent to 
\[\frac{\mu_1}{\mu_1-\lambda_1}\E[r(t_1+\Exp(\mu_1-\lambda_1))]\geq \frac{1}{\mu_1-\lambda_1}\lim_{\beta\to 0}\beta U_\beta(\alpha^*-I_1).\]

Now using Lemma~\ref{lemma:Ualpha}, the above is equivalent to
\begin{equation}
\mu_1\E[r(t_1+\Exp(\mu_1-\lambda_1))]
\geq (\mu_1-\lambda_1)\E[r(\alpha^*-I_1)]+\lambda_1\E[r(\alpha^*+\Exp(\mu_1-\lambda_1))].
\label{eq:eq42}
\end{equation}

Using Lemma~\ref{lemma:exp formula}, we have that 
\begin{align}
    &\E[r(t_1+\Exp(\mu_1-\lambda_1))]\nonumber\\&= r(\alpha^*)+\frac{1}{\mu_1-\lambda_1}\E[r'(t_1+\Exp(\mu_1-\lambda_1))], \nonumber\\
    &= r(t_1)+\frac{1}{\mu_1-\lambda_1}\left(\E[c'(t_1+\Exp(\mu_1-\lambda_1))] + \lambda_1\E[c(t_1+\Exp(\mu_1-\lambda_1))] \right) &Lemma~\ref{lemma: r'}\nonumber\\
    &=r(t_1)+\E[c(t_1+\Exp(\mu_1-\lambda_1))] - c(t_1) +  \frac{1}{\mu_1-\lambda_1}\lambda_1\E[c(t_1+\Exp(\mu_1-\lambda_1))]\nonumber\\
    &= r(t_1)+\frac{\mu_1}{\mu_1-\lambda_1}\E[c(t_1+\Exp(\mu_1-\lambda_1))] - c(t_1)\\
    &\E[r(t_1-I_1)]\nonumber=r(t_1) - \frac{1}{\lambda_1}\E[r'(t_1-I_1)]\nonumber\\
    &= r(t_1) - \frac{1}{\lambda_1}\left(\E[c'(t_1-I_1)]+\lambda_1 \E[c(t_1-I_1)] \right) &Lemma~\ref{lemma: r'} \nonumber\\
    &= r(t_1) - \left(c(t_1) - \E[c(t_1-I_1)]\right) - \E[c(t_1-I_1)] &Lemma~\ref{lemma:exp formula} \nonumber\\
    &= r(t_1)-c(t_1).
\end{align}

Note that by \eqref{eq:alpha-star}, we have that $\E[c(\alpha^*+\Exp(\mu_1-\lambda_1))]=0$. Thus  we have that 

\[\E[r(\alpha^*-I_1)]=\E[r(\alpha^*+\Exp(\mu_1-\lambda_1))].\]

Therefore, \eqref{eq:eq42} is equivalent to 
\[\E[r(t_1+\Exp(\mu_1-\lambda_1))]\geq \E[r(\alpha^*+\Exp(\mu_1-\lambda_1))].\]

This inequality holds because for any $s\geq \alpha^*$, 
\begin{align*}
    &\frac{d}{ds}\E[r(s+\Exp(\mu_1-\lambda_1))]\nonumber\\&= \E[r'(s+\Exp(\mu_1-\lambda_1))] \tag{Dominated Convergence Theorem} \\
    &= \E[c'(s+\Exp(\mu_1-\lambda_1))] + \lambda_1\E[c(s+\Exp(\mu_1-\lambda_1))] \tag{Lemma~\ref{lemma: r'}}\\
    &\geq 0 + \lambda_1\E[c(\alpha^*+\Exp(\mu_1-\lambda_1))] \tag{$c$ is non-decreasing}\\
    &= 0.
\end{align*}

\end{proof}

\subsection{Verifying Claim 2}
Finally, we use induction to prove Claim 2.
\begin{proposition}[Claim 2]
    For all $t_1\in[0, \alpha^*], t_2\geq 0$, 
    \[
        \mu_1(V(t_1, t_2) - V(t_1 - I_1, t_2)) \leq \mu_2(V(t_1, t_2)-V(t_1, t_2-I_2)).
    \]
\end{proposition}
\begin{proof}
Define function $g(t_1,t_2):=\mu_1(V(t_1, t_2) - V(t_1 - I_1, t_2)) -\mu_2(V(t_1, t_2)-V(t_1, t_2-I_2)). $ We next show that $g(t_1,t_2)\leq 0$ by  induction on $t_1$.

We know equality holds at $t_1=\alpha^*$ by Claim 1 (Proposition~\ref{claim 1}). Now assume that for some $t_1\leq \alpha^*$, we have that for all $t_2\geq 0$, $g(t_1,t_2)\leq 0$. We are going to show that there exists a $\Delta>0$ which does not depend on $t_1$, such that for any $t_1$ and $\delta\in (0,\Delta)$, we have that
    for all $t_2\geq 0$, 
    \[
        g(t_1-\delta,t_2) \leq 0.
    \]    
    We unfold the states from $(t_1-\delta, t_2)$ by a $\delta$ time-step, and use the inductive hypothesis. 
    \begin{align}
        V(t_1-\delta, t_2) &= r(t_1-\delta)\delta + (1-\mu_2\delta) V(t_1, t_2+\delta) + \mu_2\delta\ V(t_1, t_2+\delta - I_2) +o(\delta),\label{eq:V1}\\
        &= r(t_1-\delta)\delta +  V(t_1, t_2+\delta)  -\mu_2\delta ( V(t_1, t_2+\delta)- V(t_1, t_2+\delta - I_2))+o(\delta), \label{eq:V2}\\
        V(t_1-\delta-I_1, t_2) &= r(t_1-\delta-I_1)\delta + (1-\mu_2\delta) V(t_1-I_1, t_2+\delta) \nonumber\\ &\quad + \mu_2\delta\ V(t_1-I_1, t_2+\delta - I_2)+o(\delta)
        \label{eq:V3}\\
        V(t_1-\delta, t_2-I_2) &=  e^{-\lambda_2 t_2} V(t_1-\delta, -I_2)+(1-e^{-\lambda_2 t_2}) V(t_1-\delta, t_2-\bar{I}_2)+o(\delta),\nonumber\\
        &= r(t_1-\delta)\delta + 
        e^{-\lambda_2 t_2}\parens*{(1-\mu_1\delta) V(t_1, \delta-I_2) + \mu_1\delta\ V(t_1-I_1, \delta-I_2) }\nonumber\\
        &\quad +(1-e^{-\lambda_2 t_2})\Big(
        (1-\mu_2\delta) V(t_1, t_2-\bar{I}_2+\delta) \nonumber\\ &\quad + \mu_2\delta\ V(t_1, t_2-\bar{I}_2+\delta - I_2)\Big)+o(\delta),\nonumber\\
        &= r(t_1-\delta)\delta + V(t_1, t_2+\delta-I_2) \nonumber\\ &\quad -e^{-\lambda_2t_2}\cdot\mu_1\delta \parens*{V(t_1, \delta-I_2)-V(t_1-I_1, \delta - I_2)}\nonumber\\
        &\quad -(1-e^{-\lambda_2t_2}) \cdot \mu_2\delta \parens*{V(t_1, t_2+\delta-\bar{I}_2)-V(t_1, t_2+\delta -\bar{I}_2 -I_2)}+o(\delta),\nonumber\\
        &\leq r(t_1-\delta)\delta + V(t_1, t_2+\delta-I_2) \nonumber\\ &\quad -e^{-\lambda_2t_2}\cdot\mu_1\delta \parens*{V(t_1, \delta-I_2)-V(t_1-I_1, \delta - I_2)}\nonumber\\
        &\quad -(1-e^{-\lambda_2t_2}) \cdot \mu_1\delta \parens*{V(t_1, t_2+\delta-\bar{I}_2)-V(t_1-I_1, t_2+\delta -\bar{I}_2)}+o(\delta),\tag{by induction hypothesis}\nonumber\\
        &= r(t_1-\delta)\delta + V(t_1, t_2+\delta-I_2)  \nonumber\\ 
        &\quad -\mu_1\delta\parens*{V(t_1, t_2+\delta-I_2) - V(t_1-I_1, t_2+\delta-I_2)}+o(\delta).\label{eq:V4}
    \end{align}
    where $\bar{I}_2\sim [I_2 | I_2 < t_2]$. Therefore, from \Cref{eq:V1,eq:V3}, we have that
    \begin{align}
        \mu_1(V&(t_1-\delta, t_2)-  V(t_1-\delta-I_1, t_2) ) 
        = \mu_1c(t_1-\delta)\delta + \mu_1(1-\mu_2\delta)(V(t_1, t_2+\delta)\nonumber \\
        &\quad -V(t_1-I_1, t_2+\delta))+\mu_1\mu_2\delta\parens*{V(t_1, t_2+\delta-I_2)- V(t_1-I_1, t_2+\delta - I_2))}+o(\delta),\label{eq:V5}
        \end{align}
    And from \Cref{eq:V2,eq:V4}, we have that 
    \begin{align}
        &\mu_2(V(t_1-\delta, t_2)-  V(t_1-\delta, t_2-I_2)) \nonumber\\ &\geq \mu_2(V(t_1, t_2+\delta) - V(t_1, t_2+\delta-I_2))-\mu_2^2\delta(V(t_1, t_2+\delta)\nonumber\\
        &\quad - V(t_1, t_2+\delta-I_2))+\mu_1\mu_2\delta\parens*{V(t_1, t_2+\delta-I_2)- V(t_1-I_1, t_2+\delta - I_2)} +o(\delta)\\
        &= \mu_2(1-\mu_2\delta)(V(t_1, t_2+\delta) - V(t_1, t_2+\delta-I_2))\nonumber\\
        &\quad +\mu_1\mu_2\delta\parens*{V(t_1, t_2+\delta-I_2)- V(t_1-I_1, t_2+\delta - I_2))}+o(\delta).\label{eq:V6}
    \end{align}
    Combining \Cref{eq:V5,eq:V6}, 
    we have that 
    \begin{align*}
        g(t_1-\delta,t_2) \leq \mu_1 c(t_1-\delta)\delta + (1-\mu_2\delta)g(t_1,t_2) + o(\delta).
    \end{align*}

    Since we have that $\E[c(\alpha^*+\Exp(\mu_1-\lambda_1))] = 0$ and that $c$ is non-decreasing, we have that $c(t_1-\delta)\leq c(\alpha^*)<0$.\footnote{Here we can assume $c(\alpha^*)$ is strictly negative. Otherwise we have that for any $t>\alpha^*,$ $c(t)=0$. This falls back to the case when $\PPrio(2;1)$ is optimal (similar to the arguments at the beginning of Section~\ref{sec:longer-proof}).} Thus we have that 
    \begin{align*}
        g(t_1-\delta,t_2) &\leq \mu_1 c(\alpha^*)\delta + (1-\mu_2\delta)g(t_1,t_2) + o(\delta) \nonumber \\
        &\leq \mu_1 c(\alpha^*)\delta + o(\delta). && \tag{by induction hypothesis}
    \end{align*}

    Observe that the right hand side of the inequality is $o(\delta)$ plus a linear function with respect to $\delta$,  where the coefficient of the linear function is a negative constant independent of $t_1$. Therefore, we can pick $\Delta$ small enough (and independent of $t_1$) such that for any $\delta\in (0,\Delta)$, the right hand side is smaller than 0.
    
\end{proof}

\section{Simulations}\label{sec:simulations}
We now conduct simulations to evaluate the performance of our \LookAhead policy from \Cref{thm:lookahead}). We experiment with different holding cost functions, system loads and arrival rates. 

We present our results in the form of 2 experiments. The experiment correspond to a different set of holding cost, arrival and service rate parameters for the two job classes. 
In each setting, we compare our \LookAhead policy with a number of heuristic policies which have been proposed for TVHC problems. We draw this comparison across system loads, while maintaining a fixed ratio of arrival rates from each class. 

Via simulation, we demonstrate the following main findings:
\begin{itemize}
    \item Our policy is not just provably optimal, it also achieves significantly (41-56\% in \Cref{fig:1-deadline-balanced}) lower time-average holding cost than other policies.
    \item Depending on the cost functions, our policy may achieve arbitrarily lower holding cost than other heuristics. 
\end{itemize}

\subsection{Policies Evaluated}


Throughout this section, when we talk about \enquote{our policy,} we refer to our \LookAhead policy from \Cref{thm:lookahead}, where the priority of a class $i$ job of age $t$ is given by $V_i(t)$: 

\[
    V_1(t) = \mu_1\E[c_1(t+X)],  \qquad \text{where }X\sim\Exp(\mu_1-\lambda_1), \qquad \text{ and }V_2(t) = \mu_2c_2.
\]

We compare our policy against the following alternatives:
\begin{itemize}
    \item \textbf{FCFS:} This policy always serves the job that arrived earliest.  FCFS is a very simple policy and we compare against it as a baseline.
    \item \textbf{Strict Priority:} This policy assigns a fixed priority to each job class, where jobs from a higher-priority class have preemptive priority over those from a lower-priority class.  Jobs within a class are run in FCFS order. In our plots, we present the point-wise better of $\PPrio(1;2)$ and $\PPrio(2;1)$, where $\PPrio(1;2)$ denotes a preemptive priority policy where class 1 has strict priority over class 2. 
    \item \textbf{Generalized $c\mu$ Rule~\cite{van1995dynamic}:} This policy always serves the job with highest index $c_i(t)\cdot\mu_i$, where $t$ is the age of the class $i$ job. Namely, 
    \[
        V_1(t) = \mu_1c_1(t),  \qquad \text{ and }V_2(t) = \mu_2c_2.
    \]    
    This policy is known to be optimal in the diffusion limit.
    \item \textbf{Aalto's Whittle Index Policy~\cite{aalto2024whittle}:} In this policy a class $i$ job of age $t$ is given index $V_i(t)$, where  
    \[
        V_1(t) = \mu_1\E[c_1(t+S_1)],  \qquad \text{where }S_1\sim\Exp(\mu_1), \qquad \text{ and }V_2(t) = \mu_2c_2.
    \]
    This is another Whittle-based heuristic proposed for TVHC in the literature. 
    It also incorporates a \LookAhead-like intuition, however it looks ahead by a shorter time period.
\end{itemize}

Note that all the above policies, except FCFS belong to the class of $\Overtake$ policies. The strict priority policies, $\PPrio(1:2)$ and $\PPrio(2:1)$ can be respresented as $\Overtake(0)$ and $\Overtake(\infty)$ respectively. Generalized $c\mu$, Aalto and our policy each have positive finite overtake times, with Generalized $c\mu$ having the highest overtake time and our policy having the lowest. Only our policy's overtake time varies with load. Our overtake amount decreases with increasing load, which makes sense since when load increases, we want to overtake sooner so we don't end up with very high holding costs. In the limit of light load, our policy equals Aalto's policy. 

Note also that time-average total holding cost is convex in overtake age, as shown in \Cref{fig:overtake-convexity} and proved in \Cref{thm:opt-lookahead}. Our policy is at the minimum of this convex function. Aalto, generalized $c\mu$ and $\PPrio(2;1)$ are to its right, thus their costs are always ordered from lowest to highest. $\PPrio(1;2)$, being on the left of the minimum, may perform better or worse than the policies on the right.

\begin{figure}
    \centering
    \includegraphics[width=0.5\linewidth]{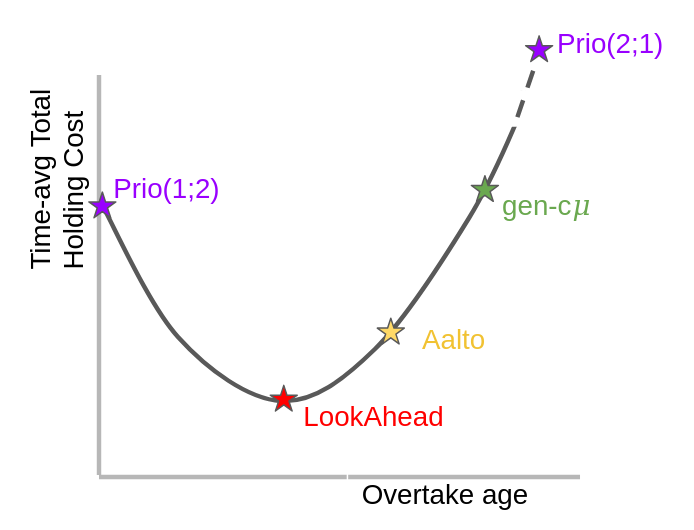}
    \caption{Expected cost is convex in overtake age. The minimum is achieved by LookAhead.}
    \label{fig:overtake-convexity}
\end{figure}



\subsection{Experimental Results}
The two experiments are chosen to highlight settings where different policies are superior.
In both settings, our policy performs significantly better than the other heuristics. Each experiment shown is represented by  {\em (a)} a set of holding cost functions (not drawn to scale); {\em (b)} the effective overtake time corresponding to each of the candidate policies (except FCFS which is not an $\Overtake$ policy);  and {\em (c)} the time-average total holding cost obtained by all policies in simulation. 

In \Cref{fig:growing}, we show an experiment with a quadratic holding cost function for class 1. As shown in \Cref{fig:growing-cost-fn}, class 1 holding cost is initially $0$ but grows quadratically to surpass class 2's holding cost.  
When we move to \Cref{fig:growing-plot}, we see that FCFS always has a very high holding cost. Strict priority (Prio) is initially the worst policy, but it becomes optimal at high load, when it is $\PPrio(1;2)$.  The generalized $c \mu$ rule and Aalto's policy both start out behaving well at lower loads, but are over 20\% worse than our policy at high loads.  Finally, in \Cref{fig:growing-errorbars}, we validate the statistical significance of the observed performance differences under high load. To account for the strong correlation in cost variation across policies (due to shared sample paths), we plot the time-average total holding cost ratio of each policy relative to ours with $2\sigma$ error-bars computed across 10 distinct sample paths. As depicted, our policy obtains consistently lower time-average total holding cost across sample paths.

We can understand the above results by looking at the {\em overtake age} under each of the above policies, as shown in \Cref{fig:growing-lookahead}. 
Recall from \Cref{fig:overtake-convexity} that the time-average holding cost is convex in the overtake age. Recall also that Aalto and the generalized $c\mu$ rule are to the right of our policy, which achieves the minimum cost. Thus it makes sense that Aalto always performs better than the generalized $c\mu$ rule in \Cref{fig:growing-plot,fig:growing-errorbars}.


Further, for quadratic holding cost functions, we can analytically compute the \LookAhead index functions, and thus the optimal overtake ages as well. The class 1 index function is quadratic in $\E[X]=\frac{1}{\mu_1-\lambda_1}$. Meanwhile the class 2 index function is constant. Thus the optimal overtake age, which is the crossing point of the 2 functions, can be derived as the root of a quadratic function. For the particular parameters in \Cref{fig:growing-cost-fn}, we have
\[
    \alpha^*(\rho) = \frac{-1}{1-0.9\rho}+\sqrt{90-\frac{1}{(1-0.9\rho)^2}}
\]
This explains the rapid decrease of the optimal overtake age. For $\rho\geq0.95$, the root of this polynomial is negative, but overtake age is constrained to be non-negative. Thus optimal overtake age is zero for $\rho\geq0.95$.





\begin{figure}[htp]
    \centering
    \begin{subfigure}[b]{0.4\textwidth}
        \centering
        \includegraphics[width=\textwidth]{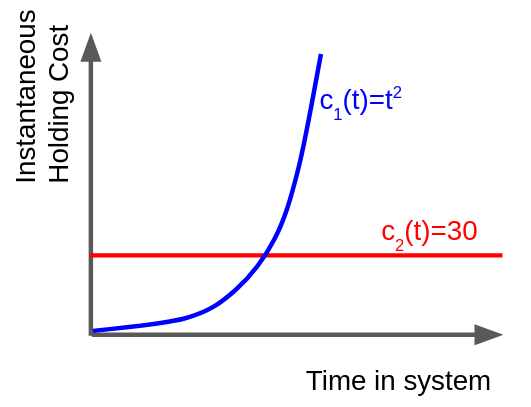}
        \caption{Holding cost functions.}
        \label{fig:growing-cost-fn}
    \end{subfigure}
    \hfill
    \begin{subfigure}[b]{0.4\textwidth}
        \centering
        \includegraphics[width=\textwidth]{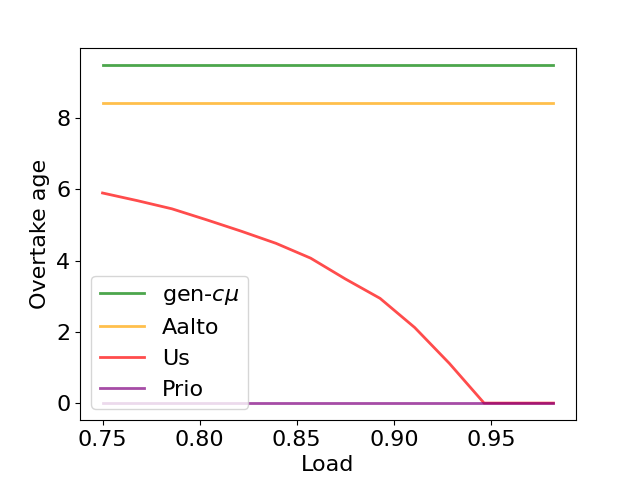}
        \caption{Overtake age.}
        \label{fig:growing-lookahead}
    \end{subfigure}


    \begin{subfigure}[b]{0.45\textwidth}
        \centering
        \includegraphics[width=\textwidth]{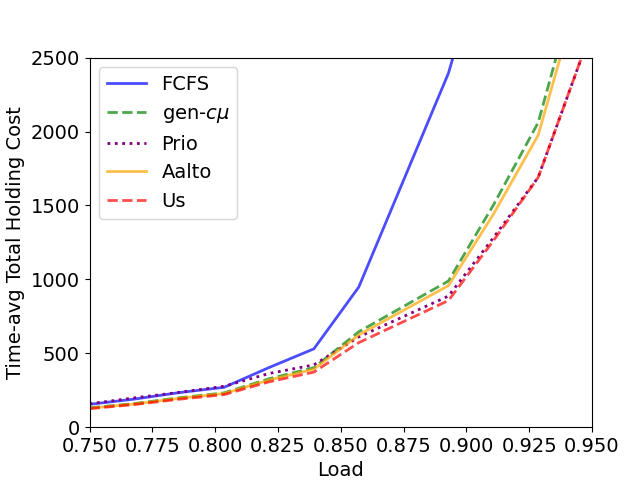}
        \caption{Performance of policies.}
        \label{fig:growing-plot}
    \end{subfigure}
    \hfill
    \begin{subfigure}[b]{0.45\textwidth}
        \centering
        \includegraphics[width=\textwidth]{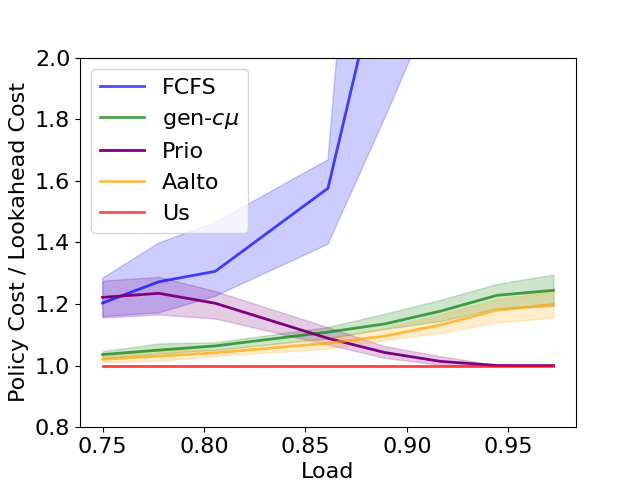}
        \caption{Cost ratio with error bars.}
        \label{fig:growing-errorbars}
    \end{subfigure}

    \caption{Comparison of policies on a setting with quadratic holding costs. We fix $\mu_1=1, \mu_2=3, \lambda_1=0.75\lambda$.}
    \label{fig:growing}
\end{figure}


In \Cref{fig:1-deadline-balanced}, we show an experiment with a deadline-based holding cost function for class 1. As shown in \Cref{fig:1-deadline-balanced-cost-fn}, a class $1$ job at age $t$ incurs  instantaneous holding cost $c_1$ past its deadline $d_1$. The corresponding time-average total holding costs obtained by the policies is shown in \Cref{fig:1-deadline-balanced-plot}. Here as well, FCFS always accrues highest holding cost. Strict priority is initially worse than generalized $c\mu$ and Aalto, but is slightly better than Aalto at high load. Our policy does significantly better than all other policies. At load 0.9, it is 56\% better than Aalto, which is the best of all other heuristics. At load 0.95, it is 41\% better than strict priority, which is the best of all other heuristics.  We plot performance with error bars in \Cref{fig:1-deadline-errorbars} to test the statistical significance of our policy's performance. Here again, we plot the cost ratio of each policy relative to ours with $2\sigma$ error-bars computed across 10 distinct sample paths. As depicted, our policy obtains consistently lower time-average total holding cost across sample paths. As depicted, the cost ratio of generalized $c\mu$ and Aalto's heuristic are higher at intermediate loads. They become closer to optimal as we approach heavy traffic in this example.

To understand what's happening, we examine the corresponding overtake ages shown in \Cref{fig:deadline-lookahead}. The optimal overtake age decreases with increasing load. However in this setting, it decreases slowly. At load $0$, the optimal overtake age is $8.87$ (the same as Aalto's overtake age). At load $0.98$ (which is the highest load we simulate), the optimal overtake time is $5.71$. In fact, in this setting with deadline-based holding costs as well, we can analytically express our LookAhead index functions. This helps us derive the optimal overtake age as a function of load to be
\[
    \alpha^*(\rho) = d_1-\ln\parens*{\frac{\mu_1c_1}{\mu_2c_2}}\cdot\frac{1}{\mu_1-\lambda_1}= 10 - \frac{\ln(30)}{3-\frac{9}{4}\rho}.
\]
Note that $\lim_{\rho\to1}\alpha^*(\rho)=10-\frac{4}{3}\ln(30) \approx 5.465$, not zero. Namely, the \LookAhead amount and optimal overtake parameter $\alpha^*$ both remain finite in the limit of heavy traffic. Thus none of the other candidate policies (generalized $c\mu$, Aalto's heuristic and strict priority) is optimal in the limit of heavy traffic. This example points out that generalized $c\mu$ is only optimal when we assume \emph{both} the diffusion limit and the limit of heavy traffic. It is not optimal simply under heavy traffic. 


\begin{figure}[htp]
    \centering
    \begin{subfigure}[b]{0.48\textwidth}
        \centering
        \includegraphics[width=\textwidth]{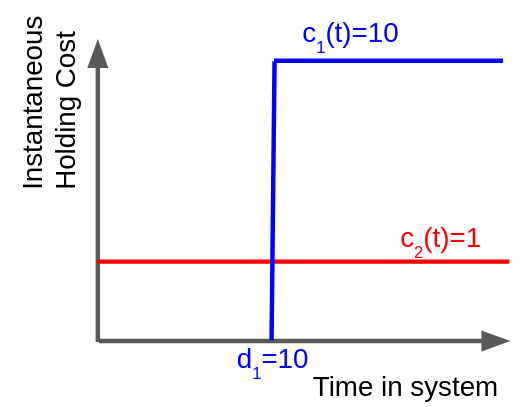}
        \caption{Holding cost functions.}
        \label{fig:1-deadline-balanced-cost-fn}
    \end{subfigure}
    \hfill
    \begin{subfigure}[b]{0.48\textwidth}
        \centering
        \includegraphics[width=\textwidth]{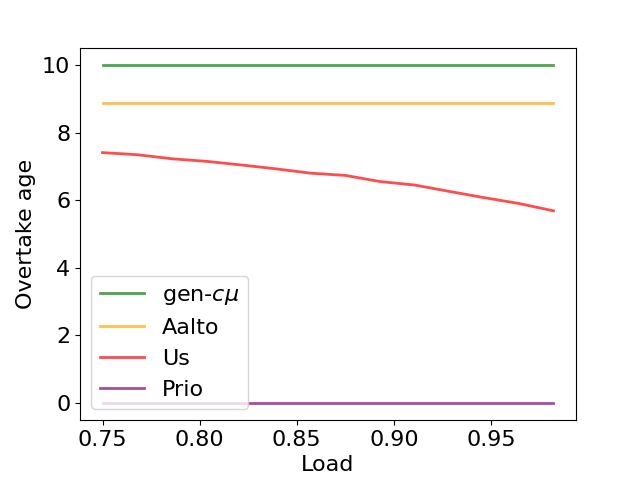}
        \caption{Overtake age.}
        \label{fig:deadline-lookahead}
    \end{subfigure}


    \begin{subfigure}[b]{0.45\textwidth}
        \centering
        \includegraphics[width=\textwidth]{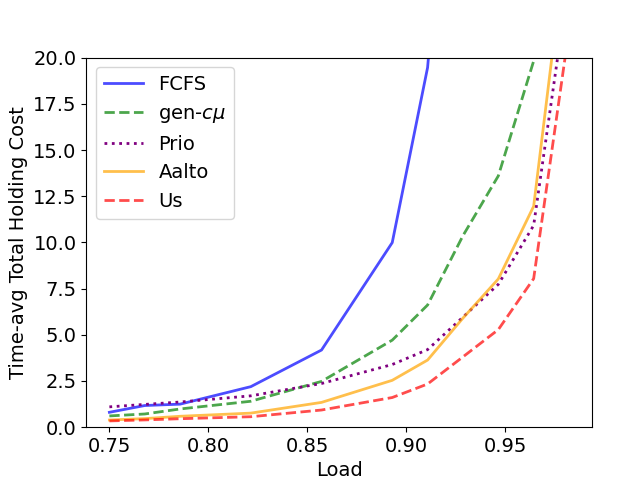}
        \caption{Performance of policies.}
        \label{fig:1-deadline-balanced-plot}
    \end{subfigure}
    \hfill
    \begin{subfigure}[b]{0.45\textwidth}
        \centering
        \includegraphics[width=\textwidth]{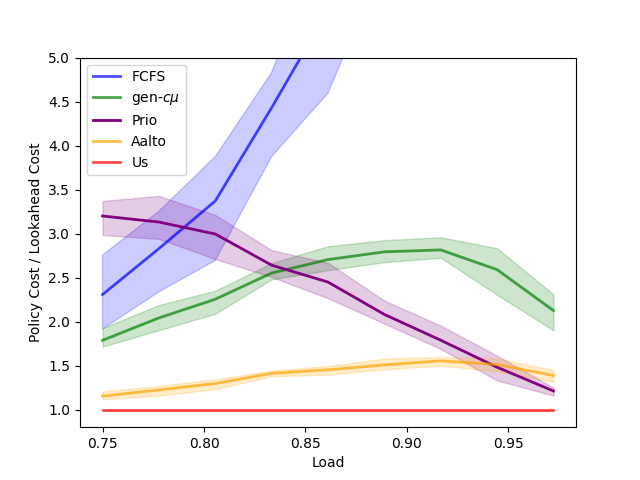}
        \caption{Cost ratio with error bars.}
        \label{fig:1-deadline-errorbars}
    \end{subfigure}

    \caption{Comparison of policies on holding cost functions with one deadline. We fix $\mu_1=3, \mu_2=1, \lambda_1=0.9\lambda$.}
    \label{fig:1-deadline-balanced}
\end{figure}

\section{Conclusion}\label{sec:conclusion}

This paper derives the first optimal scheduling policy for a TVHC problem with two classes of jobs, where one class has a holding cost that increases as jobs age. The policy derived, called \LookAhead, has a similar form to the generalized $c\mu$ rule but incorporates the holding cost of a job at a future time $X$: The index functions (for class 1 and 2) are, respectively,  $V_1(t)=\mu_1 \E[c_1(t+X)],\text{ and }V_2(t)=\mu_2 c_2$, where $X\sim \Exp(\mu_1-\lambda_1)$.  Our policy is not only optimal but also shows non-trivial improvement over existing policies.  



This work opens up promising avenues for future work.
First, the TVHC problem is an important but still widely open problem. While our work provides an optimality proof for one simple case of TVHC, we hope that our amortized holding cost intuition from Section~\ref{sec:amortized} can be useful in more general TVHC instances.
Second, there is a large community that works on multi-armed bandit problems.  Our work can be viewed as providing an optimality result within the very challenging class of {\em restless} multi-armed bandit problems.  This work can hopefully spur on the discovery of more optimality results within the class of restless multi-armed bandit problems.

\backmatter

\bmhead{Acknowledgements}

This work was supported by NSF-CIF-2403194, NSF-III-2322973, and NSF-CMMI-2307008.


\begin{appendices}

\end{appendices}

\bibliography{sn-bibliography}

\end{document}